\documentclass[useAMS,referee,usenatbib]{biom}

\usepackage{amsmath}
\usepackage{amssymb}

\usepackage{multirow}
\usepackage{booktabs}
\usepackage{setspace}
\usepackage{xcolor}

\usepackage{float}
\usepackage{rotfloat}
\floatstyle{ruled}
\newfloat{algorithm}{tpb}{loa}
\providecommand{\algorithmname}{Algorithm}
\floatname{algorithm}{\protect\algorithmname}

\title[Conformal causal inference for CRTs]
{Conformal causal inference for cluster randomized trials: model-robust inference without asymptotic approximations}

\author{Bingkai Wang$^{1,*}$\email{bingkai.w@gmail.com}, 
Fan Li$^{2,3}$, and Mengxin Yu$^{4}$ \\
$^{1}$Department of Biostatistics, University of Michigan, Ann Arbor, MI, USA\\
$^{2}$Department of Biostatistics, Yale School of Public Health, New Haven, CT, USA\\
$^{3}$Center for Methods in Implementation and Prevention Science, Yale School of Public Health,\\ New Haven, CT, USA\\
$^{4}$Department of Statistics and Data Science, University of Pennsylvania, Philadelphia, PA, USA}

\begin{document}





\pagerange{\pageref{firstpage}--\pageref{lastpage}} 




\label{firstpage}


\begin{abstract}
{Traditional statistical inference in cluster randomized trials typically invokes the asymptotic theory that requires the number of clusters to approach infinity. In this article, we propose an alternative conformal causal inference framework for analyzing cluster randomized trials that achieves the target inferential goal in finite samples without the need for asymptotic approximations. Different from traditional inference focusing on estimating the average treatment effect, our conformal causal inference aims to provide prediction intervals for the difference of counterfactual outcomes, thereby providing a new decision-making tool for clusters and individuals in the same target population. We prove that this framework is compatible with arbitrary working outcome models---including data-adaptive machine learning methods that maximally leverage information from baseline covariates, and enjoys robustness against misspecification of working outcome models. Under our conformal causal inference framework, we develop efficient computation algorithms to construct prediction intervals for treatment effects at both the cluster and individual levels, and further extend to address inferential targets defined based on pre-specified covariate subgroups. Finally, we demonstrate the properties of our methods via simulations and a real data application based on a completed cluster randomized trial for treating chronic pain.}
\end{abstract}



\maketitle


%

\section{Introduction}

Cluster randomized trials (CRTs, \citealp{murray1998design}) represent a type of randomized experiments where the treatment is assigned at the cluster level, and are prevalent in biomedical intervention studies.  
CRTs have also become increasingly popular in pragmatic clinical trials to evaluate the treatment effect in routine practice conditions. Compared to individually randomized trials, CRTs can reduce the risk of treatment contamination and facilitate the evaluation of interventions that are naturally implemented at the cluster level.

\subsection{Model-robust inference for CRTs}

{For analyzing CRTs, there is a burgeoning literature in pursuing model-robust inference in lieu of conventional model-based inference, aiming to provide stronger statistical guarantees for targeting causal estimands 
(e.g., \citealp{balzer2015adaptive,balzer2016targeted,su2021model,balzer2023two, wang2023CRT}); also see \cite{benitez2021comparative} for a recent review}.
Briefly, \emph{model robustness} means that the validity of inference does not rely on the correct specification of the working model. 
In CRTs, the true outcome may be a nonlinear function of covariates, follow a non-normal distribution, and have a complex covariate-dependent correlation structure. Due to these complexities, it is almost impossible to always pre-specify a working model that perfectly captures the true data-generating process, thereby rendering model-robust inference attractive for improving the rigor in analyzing CRTs.


Model-assisted regression and permutation tests represent two approaches for model-robust inference in CRTs, but are subject to potential limitations. Model-assisted regression requires the asymptotic approximations assuming the number of clusters (the number of independent units) approaches infinity {\citep{balzer2015adaptive, balzer2016targeted, su2021model, wang2021mixed, balzer2023two, wang2023CRT}}. 
{Often in practice, the number of randomized clusters may be too small for the asymptotic results to apply. Although ad-hoc finite-sample corrections have been recommended for small CRTs (such as the degrees-of-freedom correction in \citet{hayes2017cluster}), they have been mostly studied under model-based methods.} 
Alternatively, the permutation test can achieve model-robust exact inference under the sharp null \citep{small2008randomization, ding2018rank} and lead to conservative inference under the weak null \citep{wu2021randomization}. 
However, the permutation-based confidence interval can require non-trivial computation \citep{rabideau2021randomization}. 

\subsection{A concise introduction to conformal prediction}


{In this article, we further develop conformal prediction, originally conceptualized by \cite{vovk2005algorithmic}, to perform causal inference in CRTs. This approach achieves finite-sample validity and offers a complementary framework for model-robust inference in CRTs. We first provide a concise introduction to conformal prediction; complete details can be found in \cite{tibshirani2019conformal, barber2021predictive} among others.} 

Conformal prediction aims to construct a covariate-based interval (using an arbitrary working model) such that the outcome resides in this interval with a desired probability, e.g., 90\%. 
{Concretely, let $(Y_i,X_i)$, $i=1,\dots, n$ be independent and identical draws from an unknown distribution $\mathcal{P}$, where $Y_i$ is the outcome and $X_i$ represents a vector of covariates. Conformal prediction outputs an interval $\widehat{C}(X)$ so that $P\{Y_{\textup{new}} \in \widehat{C}(X_{\textup{new}})\} \ge 0.9$, where $(Y_{\textup{new}}, X_{\textup{new}})$ is an independent new draw from $\mathcal{P}$. Therefore, the function $\widehat{C}(X)$ gives a prediction interval for the outcome based on any new observation of covariates, ensuring that the true outcome is covered by the interval with probability at least 90\%.
Several approaches are available to construct $\widehat{C}(X)$, and their general ideas all rely on a simple fact; that is, for independent and identically distributed variables $U_1, \dots, U_n, U_{\textup{new}}$, the rank of $U_{\textup{new}}$ is uniformly distributed over $\{1,\dots, n+1\}$. Therefore, denoting $\widehat{F}$ as the empirical distribution of $U_1, \dots, U_n, U_{\textup{new}}$ and $\widehat{q}_{0.9}$ as the $0.9$-quantile of $\widehat{F}$, we have $P(U_{\textup{new}} \le \widehat{q}_{0.9}) \ge 0.9$, which leads to a naive prediction interval $(-\infty, \widehat{q}_{0.9})$ for $U_{\textup{new}}$. However, $\widehat{q}_{0.9}$ should not depend on the prediction data point $U_{\textup{new}}$. A standard approach to release this dependency is replacing $U_{\textup{new}}$ by a point mass at $+\infty$ in $\widehat{F}$, without compromising the $90\%$ coverage probability guarantee.  Under the prediction setting, $U_i$ is set to be $|Y_i-\widehat{f}(X_i)|$, where $\widehat{f}$ is an arbitrary prediction model. Then the interval $\widehat{C}(X) = \{y: |y-\widehat{f}(X)| \le \widehat{q}_{0.9}\}$ satisfies the desired coverage probability since $P\{Y_{\textup{new}} \in \widehat{C}(X_{\textup{new}})\} = P(U_{\textup{new}} \le \widehat{q}_{0.9}) \ge 0.9$. Finally, to ensure validity, data splitting can be used to separate the fitting of $\widehat{f}$ and the construction of $\widehat{C}(X)$.} 

Remarkably, this 90\% coverage probability holds even in finite samples, allowing for any arbitrary prediction model $\widehat{f}$.
Given that its model robustness is independent of any asymptotic regimes, conformal prediction addresses the potential challenge with a limited number of independent units. 
Although conformal prediction has recently been extended to perform causal inference for non-clustered data \citep{lei2021conformal,yang2022doubly, qiu2022prediction,yin2022conformal, alaa2023conformal,jin2023sensitivity}, its applications to clustered data has not been previously investigated and poses new challenges. First, the collection of all individual observations in a CRT are marginally non-exchangeable due to the intra-cluster correlation. 
As a result, the existing theory of conformal causal inference cannot be directly applied without essential modifications. Second, the multilevel data structure of CRTs implies more than one definition of the treatment effect parameter
{\citep{imai2009essential, hayes2017cluster, benitez2021comparative, kahan2023estimands,kahan2024demystifying}}, and conformal causal inference requires adaptations to target each treatment effect parameter. 
Third, existing results for conformal causal inference primarily focused on the interpretation over the entire population, but conformal causal inference about subpopulations can be of substantial interest in CRTs \citep{wang2024sample} and remains under-developed. 

\subsection{Our contribution}

Our principal contribution is to develop a conformal causal inference approach tailored for CRTs. 
Specifically, given any working model, we provide algorithms to construct prediction intervals for both the cluster-level and individual-level treatment effect and prove their finite-sample model robustness under minimal assumptions. 
The cluster-level treatment effect assesses the impact on a new cluster as a whole, while the individual-level treatment effect focuses on the impact on any given individual within a new cluster: both treatment effects can be relevant depending on the scientific question \citep{kahan2023estimands}. 
{Importantly, our prediction interval for treatment effect is different from (and not meant to replace) the traditional confidence interval for the average treatment effect. The former quantifies the uncertainty in the distribution of treatment effects over a target population, whereas the latter quantifies the uncertainty in estimating the population average treatment effect as a fixed parameter. 
The proposed conformal causal inference approach allows us to bracket treatment effects accounting for the individual- and/or cluster-specific characteristics, and represents a useful tool for individual- and cluster-based decision-making.}
Furthermore, we extend our methods to perform subgroup analysis, that is, achieving the desired coverage probability conditioning on covariate subgroups. Throughout, we integrate machine learning algorithms as a base learner for the outcome given covariates, in order to narrow the prediction intervals for treatment effects without compromising validity in finite samples.


\section{Notation, assumptions, and targets of inference}\label{sec:setup}
\subsection{Potential outcomes framework}
We consider a CRT with $m$ clusters, and each cluster $i, \ i=1,\dots, m$, includes $N_i$ individuals. At the cluster level, we define $A_i$ as a binary treatment indicator ($A_i=1$ for treatment and $A_i=0$ for control) and $R_i$ as a vector of cluster-level covariates, such as the geographical location of the cluster. For each individual $j=1,\dots, N_i$ in cluster $i$, we denote $Y_{ij}$ as their observed outcome, $X_{ij}$ as a vector of individual-level baseline covariates. The observed data for each cluster are $O_i = \{(Y_{ij}, A_i, X_{ij}, R_i):j=1,\dots, N_i\}$. We adopt the potential outcomes framework and define $Y_{ij}(a)$ as the potential outcome for individual $j$ of cluster $i$ had cluster $i$ been assigned $A_i=a$ for $a= 0,1$. We assume causal consistency such that $Y_{ij} = A_iY_{ij}(1)+(1-A_i)Y_{ij}(0)$. Denoting the complete data vector for each cluster as $W_i = \{(Y_{ij}(1), Y_{ij}(0), X_{ij}, R_i):j=1,\dots, N_i\}$, we make the following structural assumptions on $(W_1,\dots, W_m)$ and the assignment vector $(A_1,\dots, A_m)$.

\begin{assumption}[\emph{Random sampling of clusters}]\label{asmp1}
    Each $W_i$ is an independent, identically distributed sample from an unknown distribution $\mathcal{P}^W$ on $W = \{(Y_{\bullet,j}(1), Y_{\bullet,j}(0), X_{\bullet,j}, R):j=1,\dots, N\}$, where the subscript `${\bullet,j}$' represents the $j$-th individual in $W$.
\end{assumption}

\begin{assumption}[\emph{Cluster randomization}]\label{asmp2}
    Each  $A_i$ is an independent, identically distributed sample from a Bernoulli distribution $\mathcal{P}^A$ with $P(A = 1)  = \pi$ for a known constant $\pi \in (0,1)$. Furthermore, $\mathcal{P}^A$ is independent of $\mathcal{P}^W$.
\end{assumption}

\begin{assumption}[\emph{Within-cluster exchangeability}]\label{asmp3}
For any permutation map $\sigma$ on the index set $\{1,\dots, N\}$, the vector $\{W_{\bullet,1}, \dots, W_{\bullet,N}\}$ has the same distribution as $\{W_{\bullet,\sigma(1)}, \dots, W_{\bullet,\sigma(N)}\}$ conditioning on $N$, where $W_{\bullet,j} = (Y_{\bullet,j}(1), Y_{\bullet,j}(0), X_{\bullet,j}, R)$.
\end{assumption}

Assumption \ref{asmp1} 
is standard for causal inference in CRTs under a sampling-based framework \citep{wang2023CRT}. 
Assumption \ref{asmp2} holds by design. 
{Assumption \ref{asmp3} states that, given the cluster size, the joint distribution of the complete data vector in a cluster is invariant to permutation of index $j$. 
This stronger assumption is only required for inferring the individual-level treatment effect (Section \ref{sec:individual-te}) but not the cluster-level treatment effect (Section \ref{sec:cluster-te}). For additional illustration, we provide in the Supplementary Material an example class of data-generating processes that imply within-cluster exchangeability. It is clear from that example that the conditional correlation of $Y_{\bullet,j}$ and $Y_{\bullet,j'}$ is allowed to flexibly depend on $X_{\bullet,j}$ and $X_{\bullet,j'}$ through a shared function across $(j,j')$ pairs. However, Assumption \ref{asmp3} still puts stronger constraints on the data-generating distribution, compared to assumptions invoked for asymptotic model-robust methods that allow for arbitrary intracluster correlation structures; see, for example, \citet{balzer2016targeted, benitez2021comparative} and \citet{wang2023CRT}.}

\subsection{Conventional targets of inference in CRTs}
{Conventionally, statistical inference for CRTs focuses on average treatment effects defined by various summary functions \citep{su2021model, benitez2021comparative}.} For example, when the interest lies in the cluster-average treatment effect $\Delta_C=E[\overline{Y}(1)-\overline{Y}(0)]$ with $\overline{Y}(a) = \frac{1}{N} \sum_{j=1}^{N} Y_{\bullet,j}(a)$ {being the within-cluster average of potential outcomes}, one important inferential target is to construct a confidence interval $\widehat{\textup{CI}}_m$ satisfying $P\left(\Delta_C \in \widehat{\textup{CI}}_m\right) = 1-\alpha$ with $\alpha$ being a pre-specified type I error rate. Ideally, the construction of $\widehat{\textup{CI}}_m$ should further leverage covariates to reduce the type II error rate.
Despite its simple form, this objective may be difficult without parametric model assumptions. 
{To relax the dependency on parametric assumptions, asymptotic theory (corresponding to methods discussed in Section 1.1) has been developed for a weaker objective, $ \lim_{m\rightarrow \infty} P\left(\Delta_C \in \widehat{\textup{CI}}_m\right) = 1-\alpha$, which states that the target coverage probability is model-robust as the number of clusters approaches infinity.}

\subsection{Targets of conformal inference in CRTs}
For finite-sample inference that is robust to working model misspecification, we propose the following alternative inferential targets. 
Given any $\alpha \in (0,1)$, we aim to construct a conformal interval $\widehat{\textup{C}}_C(\overline{B}) \subset \mathbb{R}$ for the cluster-level treatment effect such that
\begin{equation}\label{eq:marginal-conformal-goal-C}
    P\bigg\{\overline{Y}(1) - \overline{Y}(0) \in \widehat{\textup{C}}_C(\overline{B})\bigg\} \ge 1-\alpha,
\end{equation}
where 
$\overline{B} =(\overline{X}, R, N)$ is the cluster-level covariate information, with  $\overline{X} = \frac{1}{N} \sum_{j=1}^{N_i} X_{\bullet,j}$ being the {within-cluster average of covariates}. Equation~\eqref{eq:marginal-conformal-goal-C} has the following interpretation: for any new cluster sampled from $\mathcal{P}^{W}$, its cluster-level treatment effect, $\overline{Y}(1) - \overline{Y}(0)$, is guaranteed to be contained in $\widehat{\textup{C}}_C(\overline{B})$ with probability at least $1-\alpha$, without requiring asymptotic approximations. 
Therefore, $\widehat{\textup{C}}_C(\overline{B})$ naturally quantifies the cluster-level treatment effect by characterizing the most likely range of $\overline{Y}(1) - \overline{Y}(0)$.
Distinct from the confidence interval $\widehat{\textup{CI}}_m$ that targets a fixed parameter, $\widehat{\textup{C}}_C(\overline{B})$ is covariate-based and targets the random variable $\overline{Y}(1) - \overline{Y}(0)$. 
Then, $\widehat{\textup{C}}_C(\overline{B})$ essentially provides the $(\alpha/2, 1-\alpha/2)$-quantile for the distribution of $\overline{Y}(1) - \overline{Y}(0)$.  
Of note, while the definition in \eqref{eq:marginal-conformal-goal-C} is directly motivated by the cluster-average estimand, the same framework can be directly applied to other cluster summaries such as cluster totals {\citep{su2021model, benitez2021comparative}}.

Our target of inference also extends to the construction of a conformal interval $\widehat{\textup{C}}_I(B)$ for the individual-level treatment effect such that 
\begin{equation}\label{eq:marginal-conformal-goal-I}
    P\bigg\{Y(1) - Y(0) \in \widehat{\textup{C}}_I(B)\bigg\} \ge 1-\alpha.
\end{equation}
Here, $B = (X, R, N)$ is the individual-level covariate data, and $\{Y(1), Y(0), X\}$ can represent any individual $\{Y_{\bullet,j}(1), Y_{\bullet,j}(0), X_{\bullet,j}\}$ {from a new cluster} since they are identically distributed given $N$ by Assumption \ref{asmp3}.
The interpretation of Equation~\eqref{eq:marginal-conformal-goal-I} is similar to Equation~\eqref{eq:marginal-conformal-goal-C}:  for any individual in a new cluster independently sampled from from $\mathcal{P}^{W}$, the individual-level treatment effect $Y(1)-Y(0)$ is contained in $\widehat{\textup{C}}_I(B)$ with probability at least $1-\alpha$. 
Similarly, $\widehat{\textup{C}}_I(B)$ characterizes the $(\alpha/2, 1-\alpha/2)$-quantile of $Y(1)-Y(0)$ and describes a personalized range. 
In contrast to the cluster-level treatment effect, which only pertains to independent cluster-level data, accomplishing the objectives in Equation~\eqref {eq:marginal-conformal-goal-I} is more challenging due to the unknown intracluster correlations between observations inherent in CRTs.

Beyond \eqref{eq:marginal-conformal-goal-C} and \eqref{eq:marginal-conformal-goal-I}, we further study conformal causal inference on pre-defined covariate subgroups. Specifically, let $\Omega_C$ and $\Omega_I$ be arbitrary sets with positive measures in the support of $\overline{B}$ and $B$, respectively. 
For example, if $X_{1}$ represents age, we can define $\Omega_C = \{\overline{B}: \overline{X}_1 \ge 70\}$ and $\Omega_I = \{B: X_1 \ge 70\}$ to focus on the treatment effects on a specific age group.
Given  $\Omega_C$ and  $\Omega_I$, we aim to construct conformal intervals $\widetilde{\textup{C}}_C(\overline{B})$ and $\widetilde{\textup{C}}_I(B)$ such that 

\begin{align}
P\bigg\{\overline{Y}(1) - \overline{Y}(0) \in \widetilde{\textup{C}}_C(\overline{B})\bigg|  \overline{B} \in \Omega_C\bigg\} &\ge 1-\alpha,\label{eq: local-conformal-goal-C}\\
     P\bigg\{Y(1) - Y(0) \in \widetilde{\textup{C}}_I(B)\bigg|  B\in \Omega_I\bigg\} & \ge 1-\alpha.\label{eq: local-conformal-goal-I}
\end{align}
Equations~\eqref{eq: local-conformal-goal-C} and \eqref{eq: local-conformal-goal-I} characterize ``local coverage'' that is distinct from ``marginal coverage'' characterized by Equations~\eqref{eq:marginal-conformal-goal-C} and \eqref{eq:marginal-conformal-goal-I}. 
With local coverage, the resulting conformal intervals provide more targeted information for specific sub-populations of interest (rather than the entire study population) and carry a finer resolution (in the same spirit of subgroup analyses). 
The marginal coverage definition is then viewed as a special case of local coverage by setting $\Omega_C$ and $\Omega_I$ to be the entire covariate support.


Before moving to the technical development, we introduce some additional notation. 
Let $\overline{Y} = A\overline{Y}(1) + (1-A)\overline{Y}(0)$, $I\{G\}$ denote the indicator function of event $G$, $||\cdot||_2$ be the $\ell_2$ norm,
$\delta_s$ be a point mass at $s$ in the Euclidean space, and $\delta_{+\infty}$ be a point mass at infinity.

\section{Conformal causal inference for cluster-level treatment effects}\label{sec:cluster-te}
\subsection{Inference for an observed cluster}\label{sec:cluster-te-obs}
For inferring cluster-level treatment effects, we aim to construct a conformal interval $\widetilde{\textup{C}}_C$ given a test point $\overline{B}_{\textup{test}}$, i.e., the cluster-aggregate covariates of a new cluster sampled from the target population. We first consider a basic scenario, where we observe the complete information $\overline{O}_{\textup{test}}=(\overline{Y}_{\textup{test}}, A_{\textup{test}}, \overline{B}_{\textup{test}})$ for the test point. In other words, this test point corresponds to an ``observed cluster'', i.e., its current intervention $A_{\textup{test}}$ (likely not randomized, often $A_{\textup{test}}=0$ and not included in the CRT) and current average outcome $\overline{Y}(A_{\textup{test}})$ are also recorded. Such additional data beyond $\overline{B}_{\textup{test}}$ will simplify inference since one of the two potential outcomes $\{\overline{Y}_{\textup{test}}(1), \overline{Y}_{\textup{test}}(0)\}$ is directly observed, and the conformal causal inference only requires constructing the conformal interval for the unobserved average potential outcome. 

Under this basic scenario, Algorithm~\ref{alg:1} outlines the steps to compute the conformal interval $\widetilde{\textup{C}}_C(\overline{O}_{\textup{test}})$.
In the input phase, the prediction model $f_a, a \in \{0,1\}$ can be an arbitrary map from covariates to the outcome, e.g., a linear model, or random forest \citep{breiman2001random}. Different choices of $f_a$ will not impact the coverage validity, but can result in conformal intervals with different lengths and thus affect precision. 
In addition, we need to specify the covariate subgroup of interest, $\Omega_C$, and a level $\alpha$, e.g., $\alpha=0.1$. 

\begin{algorithm}[htb]
\setstretch{1}
\caption{\label{alg:1} Computing the conformal interval $\widetilde{\textup{C}}_C(\overline{O})$ for cluster-level treatment effects.}
\textbf{Input:} Cluster-level data $\{(\overline{Y}_i, A_i, \overline{B}_i):i=1,\dots, m\}$,  a test point $\overline{O}_{\textup{test}}=(\overline{Y}_{\textup{test}}, A_{\textup{test}}, \overline{B}_{\textup{test}})$, an arbitrary prediction model $f_a(\overline{B})$ for $\overline{Y}(a)$, $a\in \{0,1\},$ a covariate subgroup of interest $\Omega_C$, and level $\alpha$.

\vspace{5pt}
\textbf{Step 1} Constructing the conformal interval $\widetilde{\textup{C}}_{C,a}(\overline{B})$ for $\overline{Y}(a)$.

For $a = 0,1$,

\hspace{5pt} 1. Randomly split the arm-$a$ covariate subgroup data $\{(\overline{Y}_i, \overline{B}_i): i=1,\dots, m, A_i =a, \overline{B}_i \in \Omega_C\}$ into a training fold $\mathcal{O}_{tr}(a)$ and a calibration fold $\mathcal{O}_{ca}(a)$ with index set $\mathcal{I}_{ca}(a)$.

\hspace{5pt} 2. Train the prediction model $f_a(\overline{B})$ using the training fold $\mathcal{O}_{tr}(a)$, and obtain the estimated model $\widehat{f}_a(\overline{B})$.

\hspace{5pt} 3. For each $i \in \mathcal{I}_{ca}(a)$, compute the non-conformity score $s(\overline{B}_i, \overline{Y}_i) = |\overline{Y}_i - \widehat{f}_a(\overline{B}_i)|$.



\hspace{5pt} 4. Compute the $1-\alpha$ quantile $\widehat{q}_{1-\alpha}(a)$ of the distribution $$\widehat{F}= \frac{1}{|\mathcal{I}_{ca}(a)|+1}\left\{\sum_{i \in \mathcal{I}_{ca}(a)}   \delta_{s(\overline{B}_i, \overline{Y}_i)} + \delta_{+\infty}\right\}.$$

\hspace{5pt} 5. Obtain $\widetilde{\textup{C}}_{C,a}(\overline{B}) = \{y \in \mathbb{R}: |y-\widehat{f}_a(\overline{B})| \le \widehat{q}_{1-\alpha}(a)\}$.

\vspace{5pt}

\textbf{Step 2} Constructing the conformal interval $\widetilde{\textup{C}}_C(\overline{O})$ for $\overline{Y}(1)-\overline{Y}(0)$.


\hspace{10pt}  If $A_{\textup{test}} = 1,$ then set $\widetilde{\textup{C}}_C(\overline{O}_{\textup{test}}) = \overline{Y}_{\textup{test}} - \widetilde{\textup{C}}_{C,0}(\overline{B}_{\textup{test}})$;

\hspace{10pt}  if $A_{\textup{test}} = 0,$ then  set $\widetilde{\textup{C}}_C(\overline{O}_{\textup{test}}) = \widetilde{\textup{C}}_{C,1}(\overline{B}_{\textup{test}}) - \overline{Y}_{\textup{test}}$.

\vspace{5pt}
\textbf{Output:} $\widetilde{\textup{C}}_C(\overline{O}_{\textup{test}})$.
\end{algorithm}



Given the input, we take two steps to construct the conformal interval. 
In the first step, we apply the split conformal prediction \citep{papadopoulos2002inductive} to construct a conformal interval for one potential outcome $\overline{Y}(a)$, based on arm-$a$ data in the covariate subgroup characterized by $\Omega_C$. In split conformal prediction, the data are randomly partitioned into two parts, one used to train the prediction model $f_a,$ $a\in \{0,1\},$ (Steps 1.1-1.2) and the other used to construct the conformal interval (Steps 1.3-1.5). 
Given the model fit $\widehat{f}_a$, Step 1.3 computes the non-conformity score, which is the absolute value of prediction error on calibration data. Intuitively,  a large non-conformity score indicates an abnormal data point (i.e., it does not ``conform'') with respect to $\widehat{f}_a$. 
{In Step 1.4, we construct an empirical distribution of the non-conformity score, denoted as $\widehat{F}$, among the validation samples and the test sample; since the test sample may not be in group $a$ and hence unobserved, we replace its point mass by a point mass at infinity. Following the same idea of conformal prediction described in Section 1.2, we have $P\{s(\overline{B}_{\textup{test}}, \overline{Y}_{\textup{test}}(a)) \le \widehat{q}_{1-\alpha}(a)\} \ge 1-\alpha$, where $\widehat{q}_{1-\alpha}(a)$ is the $(1-\alpha)$-quantile of $\widehat{F}$. This leads to our conformal interval for $\overline{Y}_{\textup{test}}(a)$ in Step 1.5.}
We then proceed to Step 2 and output $\widetilde{\textup{C}}_{C}(\overline{O}_{\textup{test}}) = (-1)^{A_{\textup{test}} + 1} \{\overline{Y}_{\textup{test}} - \widetilde{\textup{C}}_{C,1-A_{\textup{test}}}(\overline{B}_{\textup{test}})\}$; that is, the final conformal interval is a contrast between the observed potential outcome and the interval constructed for the unobserved potential outcome.

Theorem~\ref{prop:1} proves that $\widetilde{\textup{C}}_{C}(\overline{O}_{\textup{test}})$ achieves finite-sample coverage guarantee for the cluster-level treatment effect, without requiring Assumption~\ref{asmp3} on within-cluster exchangeability. 

\begin{theorem}\label{prop:1}
 Under Assumptions \ref{asmp1}-\ref{asmp2} and assuming that $\overline{O}_{\textup{test}} = (\overline{Y}_{\textup{test}}, A_{\textup{test}}, \overline{B}_{\textup{test}})$ is an independent sample from the distribution $\mathcal{P}^{\overline{O}}$ induced by $\mathcal{P}^W\times \widetilde{\mathcal{P}}^{A|W}$ and $\overline{Y}_{\textup{test}} = A_{\textup{test}}\overline{Y}_{\textup{test}}(1)+(1-A_{\textup{test}})\overline{Y}_{\textup{test}}(0)$, where $\widetilde{\mathcal{P}}^{A|W}$ is an arbitrary unknown distribution for $A_{\textup{test}}$.  Then, the conformal interval $\widetilde{\textup{C}}_{C}(\overline{O}_{\textup{test}})$ output by Algorithm~\ref{alg:1} satisfies
\begin{equation}\label{eq: prop1}
    P\bigg\{\overline{Y}_{\textup{test}}(1) - \overline{Y}_{\textup{test}}(0) \in \widetilde{\textup{C}}_C(\overline{O}_{\textup{test}})\bigg|\overline{B}_{\textup{test}} \in \Omega_C\bigg\} \ge 1-2\alpha
\end{equation}
for any set $\Omega_C$ in the support of $\overline{B}_{\textup{test}}$ with a positive measure. 

\end{theorem}

Theorem~\ref{prop:1} is a direct generalization of the conformal prediction theory to conformal causal inference on covariate subgroups. 
Since we observe $\overline{Y}_{\textup{test}}$ and $A_{\textup{test}}$, the valid coverage for the treatment effect is straightforward once we achieve valid coverage for the potential outcome. However, since we need to account for arbitrary distribution shift on $A_{\textup{test}}$ from $\mathcal{P}^A$ to $\widetilde{\mathcal{P}}^{A|W}$, the non-coverage rate $\alpha$ is doubled in Equation~\eqref{eq: prop1} by applying the union bound. If $\widetilde{P}^{A|W}$ is independent of $W$ (e.g., $P(A_{\textup{test}}=1)=0$ when studying a new intervention or $A_{\textup{test}}$ is randomized), then the coverage probability achieves $1-\alpha$ in Equation~\eqref{eq: prop1}.
{For local coverage, we consider a simple yet effective approach that only clusters within $\overline{B}_i \in \Omega_C$ are included for analysis, resembling the conventional idea for subgroup analysis.} This can change the covariate distribution but not the outcome distribution given covariates, thereby not affecting the coverage result. Beyond this approach, an alternative method that achieves uniform local coverage using the full sample is discussed in \cite{hore2023conformal}.

{In Algorithm 1, both the test point and the prediction model are based on cluster-level data. If all individual-level data are available as the test data, we could alternatively fit a prediction model $f_a^*(B_{ij})$ for $Y_{ij}(a)$ with individual-level data, and construct the non-conformity score as $|\overline{Y}_i - \frac{1}{N_{i}}\sum_{j=1}^{N_i}\widehat{f}_a^*(B_{ij})|$. With this change, one may increase the sample size for model training, thereby improving the numerical stability and the quality of the non-conformity score (See the Supplementary Material for a numerical example).}

In practice, how to choose $\alpha$ and the size of the training fold depends on the number of clusters $m$. To obtain a non-trivial conformal interval such that $\widetilde{\textup{C}}_{C}(\overline{O}_{\textup{test}}) \ne \mathbb{R}$, the definition of $\widehat{q}_{1-\alpha}(a)$ requires that $\alpha \ge |\mathcal{I}_{ca}(a)|^{-1}$. For example, if we set $\alpha=0.1$, then the calibration fold for each arm should contain at least 10 clusters, and the rest $m-20$ clusters can be used as the training fold. For stability of model fitting, we recommend using at least 20 clusters for training, while fewer clusters are also acceptable if parsimonious models such as linear regression are used. When the number of clusters is small, e.g., $m=20$, Algorithm~\ref{alg:1} based on split conformal prediction cannot construct non-trivial conformal intervals with $90\%$ coverage probability. However, this goal is achievable with full conformal prediction or Jackknife+, both of which demand substantially heavier computation \citep{barber2021predictive}.

\subsection{Inference based on cluster-level covariates}
\label{sec:cluster-te-new}

When the target of inference concerns a new cluster that only has covariate information $\overline{B}_{\textup{test}}$ (e.g., a new cluster having not taken either treatment or control studied in the current CRT), a direct approach based on Algorithm~\ref{alg:1} is to combine $\widetilde{\textup{C}}_{C,1}$ and $\widetilde{\textup{C}}_{C,0}$, for which Corollary~\ref{corollary1} characterizes its local coverage property. 

\begin{corollary}\label{corollary1}
    Under Assumptions \ref{asmp1}-\ref{asmp2} and assuming that $(\overline{Y}_{\textup{test}}(1),\overline{Y}_{\textup{test}}(0),\overline{B}_{\textup{test}})$ is an independent sample from the distribution $\mathcal{P}^W$. Then, $\widetilde{\textup{C}}_{C,1}(\overline{B}_{\textup{test}})$  and $\widetilde{\textup{C}}_{C,0}(\overline{B}_{\textup{test}})$ output by Algorithm~\ref{alg:1} satisfy
\begin{equation}\label{eq: cor1}
    P\bigg\{\overline{Y}_{\textup{test}}(1) - \overline{Y}_{\textup{test}}(0) \in \widetilde{\textup{C}}_{C,1}(\overline{B}_{\textup{test}}) - \widetilde{\textup{C}}_{C,0}(\overline{B}_{\textup{test}})\bigg|\overline{B}_{\textup{test}} \in \Omega_C\bigg\} \ge {1-2\alpha}
\end{equation}
for any set $\Omega_C$ in the support of $\overline{B}_{\textup{test}}$ with a positive measure. 
\end{corollary}

Compared to Equation~\eqref{eq: prop1}, Equation~\eqref{eq: cor1} provides {the same coverage probability $1-2\alpha$}, but the length of the conformal interval can be approximately doubled since we no longer observe $\overline{Y}_{\textup{test}}$ and $A_{\textup{test}}$. 
As a result, the conformal interval based on $\overline{B}_{\textup{test}}$ tends to be less informative than $\widetilde{\textup{C}}_C(\overline{O}_{\textup{test}})$, a natural result of less observed information. 

Beyond this direct approach,  \cite{lei2021conformal} provided a more flexible nested approach to construct $\widetilde{\textup{C}}_{C}(\overline{B}_{\textup{test}})$. This method first computes $\overline{C}_i$, the  $(1-\alpha)$-conformal interval for $\overline{Y}(1)-\overline{Y}(0)$, using  $\widetilde{\textup{C}}_{C,a}$ from Algorithm~\ref{alg:1}, and then run split conformal prediction again on $(\overline{C}_i, \overline{B}_i)$ with level $\gamma$ and prediction models $(m^L,m^R)$. Since we run split conformal prediction twice, the resulting conformal interval $\widetilde{\textup{C}}_{C}(\overline{B}_{\textup{test}})$ has non-coverage probability up to $\alpha+\gamma$. For completeness, we provide the detailed algorithm and its theoretical guarantee in the Supplementary Material. 
Finally, we remark that the nested approach becomes similar to our direct approach by setting $\gamma=\alpha$. 

\section{Conformal causal inference for individual-level treatment effects}\label{sec:individual-te}

So far, we have extended conformal causal inference for the cluster-level treatment effects by treating clusters as independent units. However, for individual-level treatment effects, the existing theory for conformal inference cannot be directly applied since the individual-level data are correlated within clusters. 
\citet{dunn2023distribution} and \citet{lee2023distribution} developed results for conformal inference with hierarchical data, focusing on the setting where individuals in the same cluster are conditionally independent. Here, we consider a more general setting where within-cluster correlation is allowed under an exchangeability setup characterized by Assumption \ref{asmp3}, which is better aligned with the characteristics of CRT data.




\subsection{Inference for an observed individual}\label{sec:individual-te-obs}

We first consider a basic setting, where the test point has the complete information $O_{\textup{test}} = (Y_{\textup{test}}, A_{\textup{test}}, B_{\textup{test}})$. This setting corresponds to an observed individual in an observed cluster of interest, whose current treatment $A_{\textup{test}}$ may not be randomized. Similar to our development in Section \ref{sec:cluster-te-obs}, we only need to construct the conformal interval for the unobserved potential outcome, likely leading to narrower and more informative conformal intervals.

In Algorithm~\ref{alg:3}, we show how to construct the conformal interval $\widetilde{\textup{C}}_I(O)$ for $Y(1)-Y(0)$. Compared to Algorithm~\ref{alg:1}, the major difference is in Step 1.4: the new empirical distribution function $\widehat{F}$ only involves individuals meeting the criteria $B_{ij} \in \Omega_I$, and each individual is further weighted by $\left(\sum_{j=1}^{N_i} I\{B_{ij} \in \Omega_I\}\right)^{-1}$; all other steps remain similar but now apply to the individual data (rather than cluster aggregate). 
{In this new $\widehat{F}$, we first construct an empirical distribution within each cluster and then construct the empirical distribution connecting validation clusters and the test cluster. Again, since the outcomes in the test cluster may not be observed, we replace its point mass with a point mass at infinity. 
This new $\widehat{F}$ captures two levels of exchangeability (within and across clusters), which extends conformal prediction of single-level data (Section 1.2 and Algorithm 1). With this change, we can establish that $P\{s({B}_{\textup{test}}, {Y}_{\textup{test}}(a)) \le \widehat{q}_{1-\alpha}(a)\} \ge 1-\alpha$ and obtain the desired conformal conformal intervals for ${Y}_{\textup{test}}(a)$; technical details are provided in the Supplementary Material.}
In addition, since we target inference on specific covariate subgroups characterized by $\Omega_I$ (which can be the entire support), we prove a Lemma in the Supplementary Material to show that Assumptions \ref{asmp3} holds when restricting to subgroup $\Omega_I$. This is a critical step to preserve the exchangeable structure within clusters and develop our theoretical results. 
The theoretical guarantee of the resulting conformal interval is formally stated in Theorem~\ref{thm1}.



\begin{algorithm}[htb]
\caption{\label{alg:3} Computing the conformal interval $\widetilde{\textup{C}}_I({O})$ for individual-level treatment effects.}
\textbf{Input:} Individual-level data $\{(Y_{ij}, A_i, B_{ij}):i=1,\dots,m; j = 1,\dots, N_i\}$,  a test point $O_{\textup{test}}=(Y_{\textup{test}}, A_{\textup{test}}, B_{\textup{test}})$, a prediction model $f_a(B)$ for $Y(a)$, $a\in \{0,1\},$ a covariate subgroup of interest $\Omega_I$, and level $\alpha$.

\vspace{5pt}
\textbf{Step 1} (Constructing the conformal interval $\widetilde{\textup{C}}_{I,a}(B)$ for $Y(a)$.)

For $a = 0,1$,

\hspace{5pt} 1. Randomly split the arm-$a$ covariate subgroup data $\{(Y_{ij}, A_i, B_{ij}):i=1,\dots,m; j = 1,\dots, N_i; A_i=a; B_{ij} \in \Omega_I\}$ into a training fold $\mathcal{O}_{tr}(a)$ and a calibration fold $\mathcal{O}_{ca}(a)$ with index set $\mathcal{I}_{ca}(a)$. {The split is at the cluster level, and individuals in the same cluster remain in the same fold.}

\hspace{5pt} 2. Train the prediction model $f_a(B)$ using the training fold $\mathcal{O}_{tr}(a)$, and obtain the estimated model $\widehat{f}_a(B)$.

\hspace{5pt} 3. For each $(i,j) \in \mathcal{I}_{ca}(a)$, compute the non-conformity score $s(B_{ij}, Y_{ij}) = |Y_{ij} - \widehat{f}_a(B_{ij})|$.

\hspace{5pt} 4. Compute the $1-\alpha$ quantile $\widehat{q}_{1-\alpha}(a)$ of the distribution $$\widehat{F}= \frac{1}{|\mathcal{I}_{ca}(a)|+1}\left\{\sum_{i \in \mathcal{I}_{ca}(a)}   \frac{1}{\sum_{j=1}^{N_i} I\{B_{ij} \in \Omega_I\}} \sum_{j=1}^{N_i}I\{B_{ij} \in \Omega_I\}\delta_{s(B_{ij}, Y_{ij})} + \delta_{+\infty}\right\}.$$

\hspace{5pt} 5. Obtain $\widetilde{\textup{C}}_{I,a}(B) = \{y \in \mathbb{R}: |y-\widehat{f}_a(B)| \le \widehat{q}_{1-\alpha}(a)\}$.

\vspace{5pt}

\textbf{Step 2} (Constructing the conformal interval $\widetilde{\textup{C}}_I(O)$ for $Y(1)-Y(0)$.)

\hspace{10pt}  If $A_{\textup{test}} = 1,$ then set $\widetilde{\textup{C}}_I(O_{\textup{test}}) = Y_{\textup{test}} - \widetilde{\textup{C}}_{I,0}(B_{\textup{test}})$;

\hspace{10pt}  if $A_{\textup{test}} = 0,$ then  set $\widetilde{\textup{C}}_I(O_{\textup{test}}) = \widetilde{\textup{C}}_{I,1}(B_{\textup{test}}) - Y_{\textup{test}}$.

\vspace{5pt}
\textbf{Output:} $\widetilde{\textup{C}}_I(O_{\textup{test}})$.
\end{algorithm}

\vspace{-0.2in}

\begin{theorem}\label{thm1}
    Under Assumptions \ref{asmp1}-\ref{asmp3} and assuming that $O_{\textup{test}}=(Y_{\textup{test}}, A_{\textup{test}}, B_{\textup{test}})$ is an arbitrary individual in a new cluster independently sampled from $\mathcal{P}^W\times \widetilde{\mathcal{P}}^{A|W}$ with $Y_{\textup{test}} = A_{\textup{test}}Y_{\textup{test}}(1) + (1-A_{\textup{test}})Y_{\textup{test}}(0)$, where $\widetilde{\mathcal{P}}^{A|W}$ is an arbitrary unknown distribution for $A_{\textup{test}}$. Then, the $\widetilde{\textup{C}}_{I}(O_{\textup{test}})$ output by Algorithm 1 satisfies
\begin{equation}\label{eq: thm1}
    P\bigg\{Y_{\textup{test}}(1) - Y_{\textup{test}}(0) \in \widetilde{\textup{C}}_I(O_{\textup{test}})\bigg|B_{\textup{test}} \in \Omega_I\bigg\} \ge 1-2\alpha
\end{equation}
for any set $\Omega_I$ in the support of $B_{\textup{test}}$ with a positive measure. 
\end{theorem}

Theorem~\ref{thm1} is the counterpart of Theorem~\ref{prop:1} for individual-level treatment effects. Due to the distribution shift on $A_{\textup{test}}$, the coverage probability is $1-2\alpha$ instead of $1-\alpha$, while the latter can be achieved if $\widetilde{\mathcal{P}}^{A|W}$ is independent of $W$ (e.g., $A_{\textup{test}}\equiv0$). When each cluster only has one individual, Algorithm~\ref{alg:1} and Algorithm~\ref{alg:3} coincide, and their resulting coverage guarantees also become identical.


In terms of the length of conformal intervals, $\widetilde{\textup{C}}_C(\overline{B}_{\textup{test}})$ tends to be more informative than $\widetilde{\textup{C}}_I(B_{\textup{test}})$ given a sufficient number of clusters. This is because $\overline{Y}(1)-\overline{Y}(0)$ often has smaller variance than $Y(1)-Y(0)$, especially if $N_i$ is large. As a result, $\widetilde{\textup{C}}_C(\overline{B}_{\textup{test}})$  is more likely to exclude zero than $\widetilde{\textup{C}}_I(B_{\textup{test}})$ given the same treatment effect size.
In practice, choosing between the two types of inferential targets requires a case-by-case evaluation. Although the scientific question should drive the target of inference, from a statistical perspective, conformal inference for the cluster-level treatment effects with $\alpha=0.1$ is typically more informative and efficient when $m$ is large, e.g., $m\ge80$. Given a small to moderate number of clusters, conformal causal inference for individual-level treatment effects could be numerically more stable due to the increased sample size in the calibration fold to compute $\widehat{q}_{1-\alpha}(a)$.

\subsection{Inference based on individual-level covariates}
When the test point only contains individual-level covariates $B_{\textup{test}}$, we follow the same strategy as in Section \ref{sec:cluster-te-new} to construct conformal intervals. Corollary~\ref{cor2} characterizes the local coverage property of the direct approach as an application of Theorem~\ref{thm1}. 
\begin{corollary}\label{cor2}
Under Assumptions \ref{asmp1}-\ref{asmp3} and assuming that $(Y_{\textup{test}}(1), Y_{\textup{test}}(0), B_{\textup{test}})$ is an arbitrary individual from a new cluster independently sampled from $\mathcal{P}^W$. Then $\widetilde{\textup{C}}_{I,1}(B_{\textup{test}})$ and $\widetilde{\textup{C}}_{I,0}(B_{\textup{test}})$ output by Algorithm~\ref{alg:3} satisfy
\begin{equation}\label{eq: cor2}
    P\bigg\{Y_{\textup{test}}(1) - Y_{\textup{test}}(0) \in \widetilde{\textup{C}}_{I,1}(B_{\textup{test}})-\widetilde{\textup{C}}_{I,0}(B_{\textup{test}})\bigg|B_{\textup{test}} \in \Omega_I\bigg\} \ge 1-{2}\alpha
\end{equation}
for any set $\Omega_I$ in the support of $B_{\textup{test}}$ with a positive measure. 
\end{corollary}


In parallel to Corollary~\ref{corollary1}, Corollary~\ref{cor2} establishes the coverage guarantee on conformal intervals for the individual-level treatment effect based on covariates. In addition, the resulting conformal interval enjoys the same benefit of stability and small-sample compatibility as discussed in Section \ref{sec:individual-te-obs}. 
In addition to the direct approach, we provided the nested approach with $1-\alpha-\gamma$ coverage guarantee in the Supplementary Material.

\vspace{-0.2in}

\section{Simulations}\label{sec:simulation}
\subsection{Simulation design}
Through a simulation study, we demonstrate our finite-sample theoretical results for both cluster-level and individual-level treatment effects. We consider the combination of the following settings: CRTs with a large ($m=100$) versus small ($m=30$) number of clusters, and full-data analysis versus covariate subgroup analysis. For $i=1,\dots, m$, we independently generate the cluster size $N_i \sim \mathcal{U}([10,50])$ and two cluster-level covariates $R_{i1}|N_i \sim \mathcal{N}(N_i/10, 1)$, $R_{i2}|(N_i, R_{i1}) \sim \mathcal{B}\{(1+e^{-R_{i1}/2})^{-1}\}$, where $\mathcal{U}, \mathcal{N}, \mathcal{B}$ represent the uniform, normal, and Bernoulli distribution, respectively. 
For each individual $j = 1,\dots, N_i$, we independently generate covariates $X_{ij1}|(N_i, R_{i1}, R_{i2}) \sim \mathcal{B}(0.3+0.4R_{i2})$ and $X_{ij2} = (2I\{R_{i1}>0\}-1)\overline{X}_{i1} + \mathcal{N}(0,1)$, and potential outcomes $Y_{ij}(a)= aN_i/50+\sin(R_{i1}) (2R_{i2}-1) + |X_{ij1}X_{ij2}| + (1-a)\gamma_i + \varepsilon_{ij}$ for $a=0,1$
where $\gamma_i \sim \mathcal{N}(0,0.5^2)$ is the random intercept and $\varepsilon_{ij} \sim N(0,1)$ is the random noise, {leading to an adjusted intracluster correlation coefficient of $0.2$ under $a=0$}.  Then we independently generate the treatment indicator $A_i \sim \mathcal{B}(0.5)$, and obtain $Y_{ij} = A_i Y_{ij}(1) + (1-A_i)Y_{ij}(0)$. The simulated observed data are $\{(Y_{ij}, A_i, X_{ij1}, X_{ij2}, R_{i1}, R_{i2}): i=1,\dots, N_i\}_{i=1}^m$. 
{To compute performance metrics}, we generate 1,000 new clusters as the test data, following the same data-generating distribution, and repeat the above procedure to generate 1,000 data replicates.

For each simulated data set, we first construct the conformal interval with $\alpha=0.1$ for the cluster-average treatment effect. Given complete test data (with treatment and outcomes), we run Algorithm~\ref{alg:1}, and refer to this approach as ``O''. Given cluster-level covariates only, we refer to the direct approach as ``B-direct'' and the nested approach as ``B-nested''. For covariate subgroup analysis, we consider $\Omega_C = \{\overline{B}_i: R_{i1}\ge2, R_{i2}=1\}$, which contains 60\% of all clusters. For inferring the individual-average treatment effect, we adopt the same names ``O'', ''B-direct'', ``B-nested'' to refer to the output of Algorithm~\ref{alg:3}, the direct approach and the nested approach, and the covariate subgroup is defined as $\Omega_I = \{B_{ij}: |X_{ij2}|<0.5\}$, which includes $30\%$ of all individuals. {While our theoretical results support any choice of training models, to improve predictive accuracy, we consider the training model to be an ensemble learner of linear regression and random forest implemented via the \texttt{SuperLearner} R package \citep{van2007super}.} 
We consider two metrics of performance based on test data: the probability that the conformal interval contains the true treatment effect, and the average length of the conformal interval. 
For the nested approach, we set $\gamma=0.5$ to improve the informativeness of the resulting conformal interval. 

\subsection{Simulation results}
Figure~\ref{fig:sim-1} summarizes the simulation results for the marginal and local cluster-level treatment effects given $m=100$. In Figure~\ref{fig:sim-1}, the upper panels show that all three methods achieve the target 90\% coverage probability (reflected by the medians of all box plots sitting above 0.9), thereby confirming our theoretical results. Comparing the three methods, the ``O'' method has a coverage probability close to 0.9, whereas the ``B-direct'' method is the most conservative (coverage probability near 1). This difference can be explained by the lower panels, where the ``B-direct'' method yields wider conformal intervals, whose length nearly doubles the oracle length. In contrast, due to leveraging the complete information in the test data, the ``O'' method achieves the near-optimal length of the conformal interval. The performance of the ``B-nested'' method lies between the other two since we set a loose parameter $\gamma=0.5$. If we use $\gamma=0.1$ instead, it will perform similarly to the ``B-direct'' method as demonstrated in \cite{lei2021conformal} for non-clustered data. 
Comparing results for marginal versus local treatment effects, the coverage probability and length of conformal intervals are generally similar. 
In the Supplementary Material, we reproduce Figure~\ref{fig:sim-1} with $m=500$, where the span of the boxplot substantially decreases and the ``O'' method is nearly optimal under both metrics. 

\begin{figure}[htbp]
    \centering
    \includegraphics[width=0.95\textwidth]{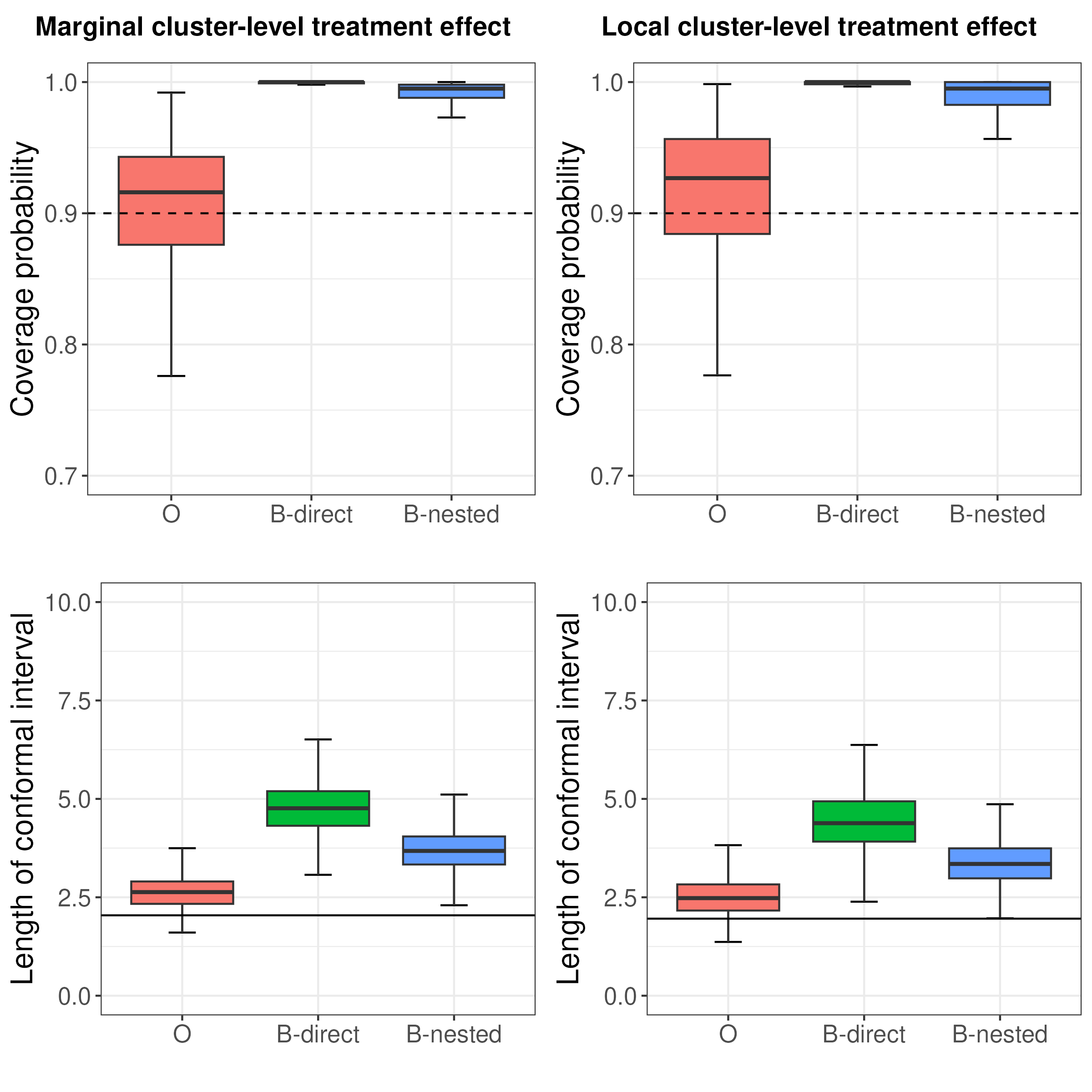}
    \caption{Simulation results (boxplot) for the marginal (left column) and local (right column, conditioning on $\{R_{i1}\ge 2, R_{i2}=1\}$) cluster-level treatment effects with $m=100$. In the upper panels, the dashed line is the target 90\% coverage probability. In the lower panels, the solid line is the oracle length of conformal intervals, computed as the average length between the $(\alpha/2,1-\alpha/2)$-quantiles of $\overline{Y}(1)-\overline{Y}(0)$ among test data. {Each box plot is based on 1000 data points, with each representing the performance metric computed from one data replicate (instead of one test data point).}}
    \label{fig:sim-1}
\end{figure}

Figure~\ref{fig:sim-2} presents the results for the marginal and local individual-level treatment effects given $m=100$. All three methods reach the target coverage probability, with patterns similar to Figure~\ref{fig:sim-1}. However, the conformal intervals for the individual-level treatment effect appear more stable than the cluster-level treatment effect, as reflected by shorter spans of the boxplot, but have wider lengths due to the increased variance in treatment effects. Results for $m=500$ are presented in the Supplementary Material with similar patterns. 

\begin{figure}[htbp]
    \centering
    \includegraphics[width=0.95\textwidth]{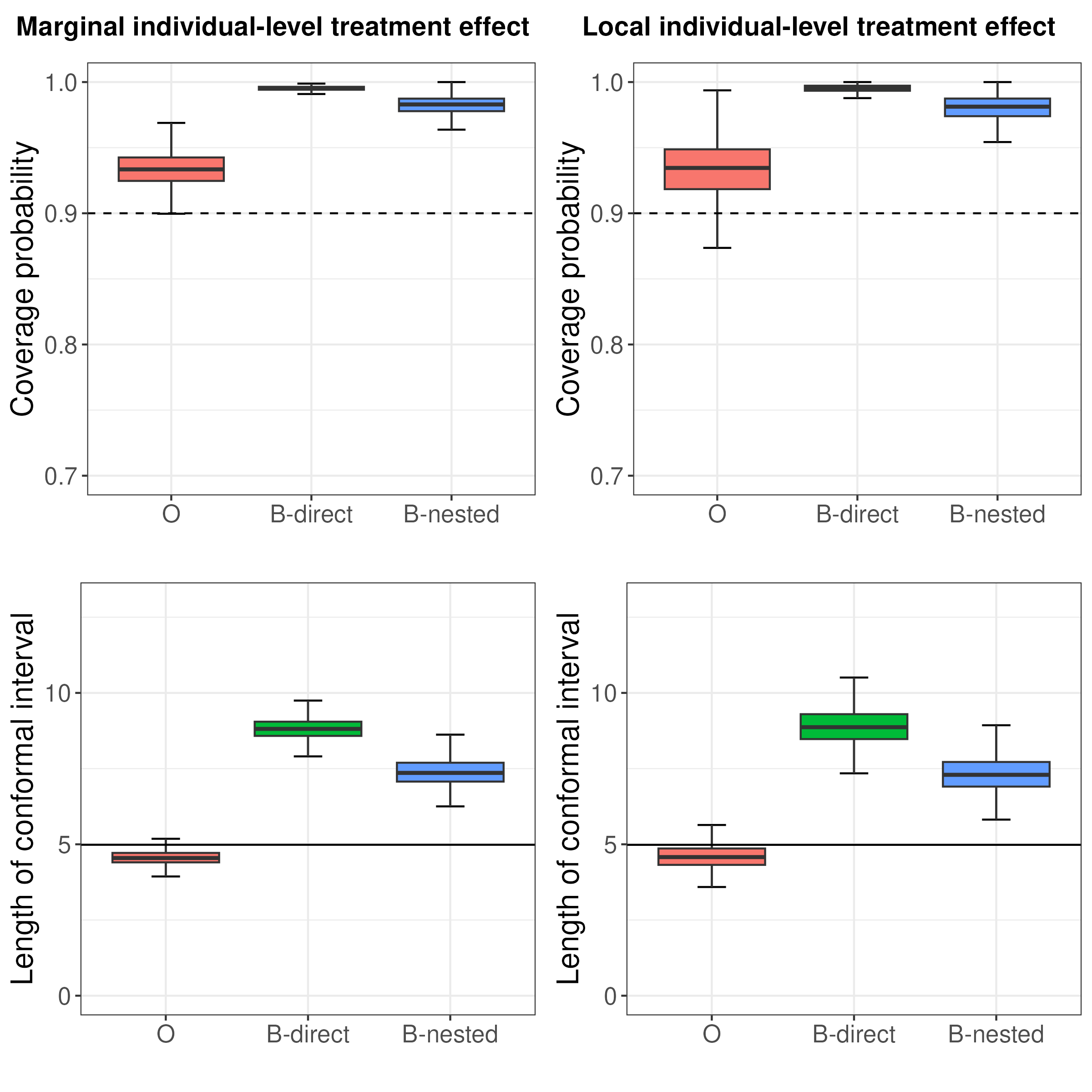}
    \caption{Simulation results (boxplot) for the marginal (left column) and local (right column, conditioning on $\{|X_{ij2}|<0.5\}$) individual-level treatment effects with $m=100$. In the upper panels, the dashed line is the target 90\% coverage probability. In the lower panels, the solid line is the oracle length of conformal intervals, computed as the average length between the $(\alpha/2,1-\alpha/2)$-quantiles of $Y(1)-Y(0)$ among test data.{Each box plot is based on 1000 data points, with each representing the performance metric computed from one data replicate (instead of one test data point).}}
    \label{fig:sim-2}
\end{figure}

In the Supplementary Material, we provide additional simulations under $m=30$ to test the performance under a smaller number of clusters. Under this scenario, we still observe valid coverage probability by the ``O'' and ``B-direct'' methods, but they become more conservative, as reflected by the increased length of intervals. In addition, we repeat our simulations for $m=30$ with only linear regression as the training model (in contrast to the ensemble method). We find that, by including random forest in the ensemble method for model training, the coverage probability for the conformal intervals has negligible differences, but the length of intervals is reduced by 8.2-46.5\%. This example demonstrates the improved accuracy in conformal causal inference by leveraging data-adaptive machine learners.


\section{Data example with the PPACT cluster randomized trial}\label{sec:data-application}
The Pain Program for Active Coping and Training study (PPACT) is a CRT evaluating the effect of a care-based cognitive behavioral therapy intervention for treating long-term opioid users with chronic pain \citep{debar2022primary}. 
The study equally randomized 106 primary care providers (clusters) to receive the intervention or usual care, with 1-10 participants in each cluster. 
We focus on the primary outcome, the PEGS (pain intensity and interference with enjoyment of life, general activity, and sleep) score at 12 months, a continuous measure of pain scale ranging from 1 to 10. For more accurate conformal causal inference, we adjust for 13 individual-level baseline variables, including the baseline PEGS score, age, gender, disability, smoking status,
body mass index, alcohol abuse, drug abuse, comorbidity, depression, number of pain types, average morphine dose, and heavy opioid usage. 

{In the real-world setting, we can use all 106 clusters to construct conformal interval functions $\widetilde{\textup{C}}_C(\overline{O})$ and $\widetilde{\textup{C}}_I(O)$ that are applicable to any new cluster or individuals in the new cluster by plugging in $\overline{O}_{\textup{test}}$ and $O_{\textup{test}}$. Here, to demonstrate our approaches, we randomly sample 20 clusters as the test data to compute performance metrics and use the rest 86 clusters to construct conformal interval functions; we repeat this process for 100 times to account for the uncertainty in the data split. }
We report the average and standard error for two performance metrics: length of intervals and fraction of negatives. Here, the fraction of negatives defines the proportion of conformal intervals that are subsets of $(-\infty, 0)$ among the test data. Since negative values are in the direction of treatment benefits, this metric reveals how many clusters/individuals are associated with beneficial treatment effects with probability $1-\alpha$, {and bears a similar interpretation to power for a one-sided test}. Because the test data have the complete information, we directly run the ``O'' method, i.e., Algorithm~\ref{alg:1} and Algorithm~\ref{alg:3}, with $f_a$ set as the ensemble learner of linear regression and random forest.

Table~\ref{data-application1} summarizes the results for the marginal treatment effects setting $\alpha\in\{0.1,0.2,0.3,0.4\}$. As $\alpha$ increases, the length of intervals decreases, and the fraction of negative becomes larger. Since the treatment effect is small (relative to the variability of treatment effects,  \citealp{wang2023CRT}), only a small to moderate proportion of the population has negative conformal intervals. In practice, these negative conformal intervals can be informative for new patients and clusters from the same source population generating the observed trial sample. 

\begin{table}[htbp]
\caption{Summary results of data application for marginal treatment effects. For both the length of intervals and the fraction of negatives, we present the average and standard error over 100 runs.}\label{data-application1}
\centering
\resizebox{\textwidth}{!}{
\begin{tabular}{ccccc}
  \hline
\multirow{2}[3]{*}{\shortstack{Coverage\\ probability}}  & \multicolumn{2}{c}{Marginal cluster-level treatment effect} & \multicolumn{2}{c}{Marginal individual-level treatment effect}\\
  \cmidrule(lr){2-3} \cmidrule(lr){4-5}
 & Length of intervals & Fraction of negatives & Length of intervals & Fraction of negatives \\ 
  \hline
90\% & 4.056(0.557) & 0.089(0.069) & 6.908(0.590) & 0.055(0.026) \\ 
  80\% & 2.874(0.412) & 0.173(0.092) & 4.800(0.345) & 0.132(0.044) \\ 
  70\% & 2.233(0.313) & 0.238(0.104) & 3.811(0.220) & 0.199(0.052) \\ 
  60\% & 1.799(0.255) & 0.304(0.117) & 3.108(0.189) & 0.257(0.055) \\ 
   \hline
\end{tabular}
}
\end{table}

{If the research interest lies in obtaining the conformal interval for a cluster or an individual within the CRT sample, we could use the target cluster as the test data and the rest as the training and calibration data. As a demonstration, Figure~\ref{fig:data-application-125} gives the cluster-level and individual-level conformal intervals with 90\% coverage for a control cluster. It shows that the treatment is 90\% likely to be beneficial at the cluster level; at the individual level, one of the three conformal intervals excludes zero, and the treatment is 90\% likely to be beneficial for this individual. }

\begin{figure}[htbp]
\centering
        \includegraphics[width=0.8\textwidth]{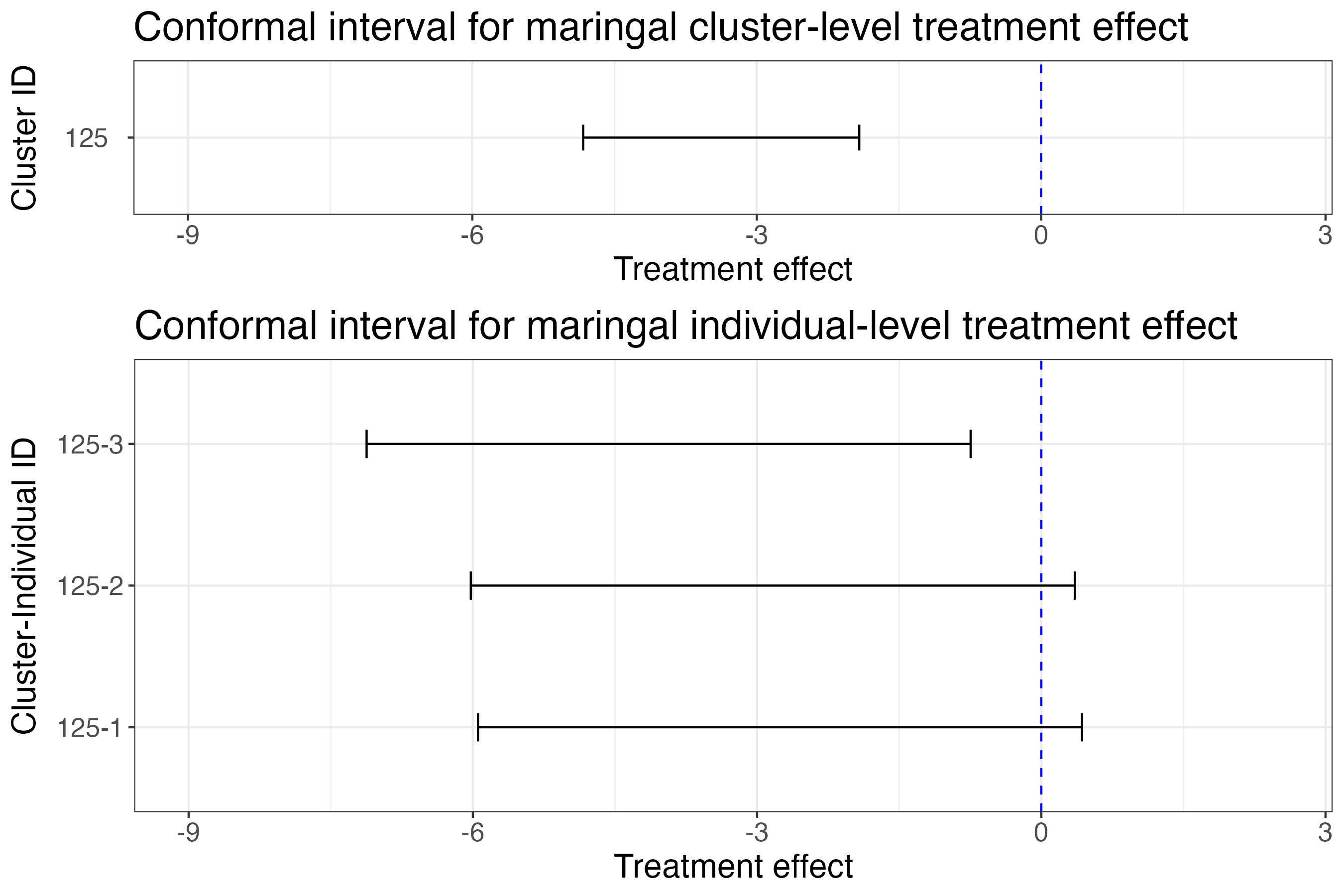}
        \caption{Cluster-level and individual-level treatment effects for one cluster with three individuals in the controlled arm of the PPACT study.}
\label{fig:data-application-125}
\end{figure}

Finally, we performed subgroup analyses on individuals with severe or moderate baseline pain, and the results are summarized in the Supplementary Material.

\section{Discussion}\label{sec:discussion}


{In this article, we develop the conformal causal inference framework to study treatment effects in CRTs, and offer a complementary non-asymptotic framework to existing approaches that target the average treatment effects \citep[e.g.,][]{benitez2021comparative,su2021model,wang2023CRT}.} Although our data example in Section \ref{sec:data-application} includes a large number of clusters, 
our proposed methods are not confined to a large number of clusters, given appropriate adjustments of the target coverage probability. 
As a practical consideration, for achieving 90\% coverage probability of the conformal interval for the cluster-level treatment effects, we need at least 10 clusters per arm for numerically stable calibration, and our simulation shows that the conformal interval is valid but moderately conservative given $m=30$ and $\pi = 0.5$. Given fewer clusters, e.g., $m=20$, targeting 90\% coverage probability will only result in intervals spanning the entire real line, and more informative intervals are possible with a lower target coverage probability (say 80\%). In practice, we recommend setting $\alpha=0.05$ if $m\ge 80$, $\alpha=0.1$ if $40\le m <80$, and $\alpha=0.2$ if $20\le m < 40$ given $\pi = 0.5$. If $10< m < 20$, we can still construct the 80\% conformal interval for the individual-level treatment effect, but it may be challenging to make conformal causal inference given $m\le 10$. Conformal causal inference with an extremely small number of clusters is a topic of future research. 

{We have considered the basic CRT setting in this development, but our results can be extended in several directions.
First, 
although we assume cluster randomization, conformal causal inference can be extended to clustered observational studies with a cluster-level treatment under strong ignorability, along the lines of \citet{lei2021conformal} and \citet{yang2022doubly}. With this change, one can perform conformal causal inference for treatment effects among the treated clusters and individuals, and the details are provided in the Supplementary Material. 
Second, we assume simple cluster randomization, whereas in other cases covariate-adaptive randomization, such as pair-matching \citep{balzer2016adaptive}, may be used. The extension of conformal causal inference to accommodate covariate-adaptive randomization represents an important future research direction. Lastly, it would also be useful to relax the within-cluster exchangeability assumption for inferring individual-level treatment effects. Although some efforts have been made in the conformal prediction literature to address the issues of non-exchangeability and distribution shifts under various settings \citep{tibshirani2019conformal,barber2023conformal,dobriban2023symmpi}, addressing non-exchangeability with clustered data presents non-trivial challenges and warrants additional investigation. }





\backmatter


\section*{Acknowledgements}
Research in this article was supported by a Patient-Centered Outcomes Research Institute Award\textsuperscript{\textregistered} (PCORI\textsuperscript{\textregistered} Award ME-2022C2-27676) and the National Institute Of Allergy And Infectious Diseases of the National Institutes of Health under Award Number R00AI173395. The statements presented in this article are solely the responsibility of the authors and do not necessarily represent the official views of the National Institutes of Health or PCORI\textsuperscript{\textregistered}, its Board of Governors, or the Methodology Committee.

\vspace{-0.1in}
\section*{Supplementary Material}
Web Appendices, Tables, and Figures, and code referenced in Sections~\ref{sec:cluster-te}-\ref{sec:data-application} are available with this paper at the Biometrics website on Oxford Academic.

\vspace{-0.1in}
\section*{Data availability}
The data underlying this article were provided by the PPACT study team. Data will be shared on request to the corresponding author with permission of \cite{debar2022primary}.

\vspace{-0.1in}


%
\bibliographystyle{biom} 
\bibliography{ref}








\label{lastpage}

\end{document}


\def\spacingset#1{\renewcommand{\baselinestretch}%
{#1}\small\normalsize} \spacingset{1}

\date{\vspace{-5ex}}

\maketitle

\renewcommand\thesection{\Alph{section}}
\renewcommand\thealgorithm{\Alph{section}.\arabic{algorithm}}
\renewcommand\thetheorem{\Alph{section}.\arabic{theorem}}

\spacingset{1.5}

In Section~\ref{sec: exchangeability-examples}, we provide an example class of data generating distributions that imply within-cluster echangeability.
In Section~\ref{sec:proofs}, we provide the proofs for our theoretical results. In Section~\ref{sec: nested}, we provide the nested approach for constructing conformal intervals based on CRTs. In Section~\ref{sec: additional-simulations}, we provide additional numerical results. 
In Section~\ref{sec: cluster-observational-studies}, we show how the proposed approach can be extended to study the treatment effect on the treated in cluster observational studies.
The R code for reproducing all simulation and data analysis is available at \url{https://github.com/BingkaiWang/CRT-conformal}.

\section{A class of data-generating processes that imply Assumption 3}\label{sec: exchangeability-examples}

To contextualize the within-cluster exchangeability assumption (required for inferring the individual-level treatment effect target), we discuss am example class of data-generating processes for which Assumption 3 holds. This example is meant to illustrate that, under Assumption 3, the unknown intracluster correlation structure among the outcomes can still flexibly depend on the covariates, and the model space after assuming within-cluster exchangeability remains considerably large.

We first assume that the joint distribution of $\{X_{\bullet,1}, \ldots, X_{\bullet,N}\}$ given $R$, $N$ is invariant to permutation of individual component indices. A necessary condition for this process is that the marginal distribution of $X_{\bullet,j}$ is common across $j$. This assumption, although strong, is often invoked in simulation studies for validating and comparing methods for CRTs. Next, we consider the following restricted moment random coefficient model for the pair of individual potential outcomes:
\begin{align}\label{eq:bivariate}
    \left(\begin{array}{c}
        Y_{\bullet,j}(1) \\
        Y_{\bullet,j}(0) \end{array}\right)&=
    \left(\begin{array}{c}
        \eta_1(X_{\bullet,j},R,N)\\
        \eta_0(X_{\bullet,j},R,N)\end{array}\right)
    + \left(\begin{array}{c}
        g_1^\top(X_{\bullet,j})\gamma(1)\\
        g_0^\top(X_{\bullet,j})\gamma(0)\end{array}\right)
    +\left(\begin{array}{c}
       {h_1(X_{\bullet,j}, R, N)} \epsilon_{\bullet,j}(1)\\
        {h_0(X_{\bullet,j}, R, N)}\epsilon_{\bullet,j}(0)\end{array}\right),
\end{align}
where $\eta_1$, $\eta_2$ are completely unspecified functions of the baseline covariates and represent fixed effects, $g_1(X_{\bullet,j})$, $g_0(X_{\bullet,j})$ are design vectors (including intercept) for random intercept and coefficients and are an unknown function of $X_{\bullet,j}$, $(\gamma^\top(1),\gamma^\top(0))^\top$ is a random effect vector that jointly follows a multivariate distribution with mean zero and variance-covariance matrix $\Lambda(R,N)=\begin{pmatrix}
\Lambda_1(R,N) & \Lambda_{10}(R,N)\\
\Lambda_{10}^\top(R,N) & \Lambda_0(R,N)\\
\end{pmatrix}$ that is allowed to depend on cluster-level covariates $R,N$, and independently, $(\epsilon_{\bullet,j}(1),\epsilon_{\bullet,j}(0))^\top$ are bivariate errors that follow a bivariate distribution with mean zero and variance-covariance matrix {$\Sigma=\begin{pmatrix}
\sigma_1^2 & \sigma_{10}\\
\sigma_{10} & \sigma_{0}^2\\
\end{pmatrix}$ and that are also independent across $j$.}
$\Sigma(X_{\bullet,j})=\begin{pmatrix}
\sigma_1^2(X_{\bullet,j}) & \sigma_{10}(X_{\bullet,j})\\
\sigma_{10}(X_{\bullet,j}) & \sigma_{0}^2(X_{\bullet,j})\\
\end{pmatrix}$ that is allowed to depend on individual-level covariates. 
This data-generating process has the following features. (I) \emph{semiparametric}: the class of data-generating processes only impose moment conditions for the cluster-level random effects (intercept and coefficients) and the individual-level random errors, but leaves the fixed-effects components and random-effects design vectors completely unspecified; (II) \emph{heteroscedastic}: the variance-covariance matrix of the cluster-level random effects is allowed to depend on cluster-level covariates, and the variance-covariance matrix of the individual-level errors is allowed to depend on individual-level covariates, without specifying the form of dependency. When $g_1(X_{\bullet,j})=g_0(X_{\bullet,j})=1$, and the distributions of $(\gamma(1),\gamma(0))^\top$, $(\epsilon_{\bullet,j}(1),\epsilon_{\bullet,j}(0))^\top$ are both bivariate normal with constant variance-covariance matrices, \eqref{eq:bivariate} becomes the bivariate linear mixed model discussed, for example, in \citet{yang2023power} for CRTs with multivariate outcomes; (III) \emph{heterogeneous intracluster correlation structure}: the class of models \eqref{eq:bivariate} imply the following intracluster correlation coefficients (ICCs) {that can vary by pairs of individuals and by clusters, depending on the individual-level and cluster-level covariates}. 
\begin{enumerate}
\item For two treated (control) potential outcomes from different individuals in the same cluster, the \emph{single-world between-individual ICC} is given by
\begin{align*}
&Corr(Y_{\bullet,j}(1),Y_{\bullet,l}(1)|X_{\bullet,j},X_{\bullet,l}{, R, N})=\frac{g_1^\top(X_{\bullet,j})\Lambda_1(R,N)g_1(X_{\bullet,l})}{\sqrt{V_{\bullet,j}(1)}\sqrt{V_{\bullet,l}(1)}},\\
&Corr(Y_{\bullet,j}(0),Y_{\bullet,l}(0)|X_{\bullet,j},X_{\bullet,l}{, R, N})=\frac{g_0^\top(X_{\bullet,j})\Lambda_0(R,N)g_0(X_{\bullet,l})}{\sqrt{V_{\bullet,j}(0)}\sqrt{V_{\bullet,l}(0)}},
\end{align*}
where $V_{\bullet,j}(a)=g_0^\top(X_{\bullet,j})\Lambda_a(R,N)g_a(X_{\bullet,j})+{h_a^2(X_{\bullet,j}, R, N)}\sigma_a^2$ for $a =0,1$.
\item For each individual, the \emph{cross-world within-individual ICC} for the pair of potential outcomes is 
\begin{align*}
&Corr(Y_{\bullet,j}(1),Y_{\bullet,j}(0)|X_{\bullet,j}{, R, N})=\frac{g_1^\top(X_{\bullet,j})\Lambda_{10}(R,N)g_0(X_{\bullet,j})+{h_1(X_{\bullet,j}, R, N)}{h_0(X_{\bullet,j}, R, N)}\sigma_{10}}{\sqrt{V_{\bullet,j}(1)}\sqrt{V_{\bullet,j}(0)}}.
\end{align*}
\item For potential outcomes from two different individuals in the same cluster, the \emph{cross-world between-individual ICC} is 
\begin{align*}
&Corr(Y_{\bullet,j}(1),Y_{\bullet,l}(0)|X_{\bullet,j},X_{\bullet,l}{, R, N})=\frac{g_1^\top(X_{\bullet,j})\Lambda_{10}(R,N)g_0(X_{\bullet,l})}{\sqrt{V_{\bullet,j}(1)}\sqrt{V_{\bullet,l}(0)}}.
\end{align*}
\end{enumerate}

Next, we show that within-cluster exchangeability holds under this class of semiparametric data-generating processes. To see this, we factorize the joint distribution of the full data vector for a cluster by 
\begin{align*}
f(W_{\bullet,1}, \dots, W_{\bullet,N})=&\left\{\prod_{j=1}^N f(Y_{\bullet,j}(1),Y_{\bullet,j}(0)|\gamma(1),\gamma(0),X_{\bullet,j},R,N)\right\}\\
&\times f(\gamma(1),\gamma(0)|R,N) f(X_{\bullet,1}, \ldots, X_{\bullet,N}|R,N)f(R,N).
\end{align*}
Under any random permutation of indices (mapping from $\{1,\dots,N\}$ to $\{\sigma(1),\dots,\sigma(N)\}$), we get
\begin{align*}
f(W_{\bullet,\sigma(1)}, \dots, W_{\bullet,\sigma(N)})=&\left\{\prod_{j=1}^N f(Y_{\bullet,\sigma(j)}(1),Y_{\bullet,\sigma(j)}(0)|\gamma(1),\gamma(0),X_{\bullet,\sigma(j)},R,N)\right\}\\
&\times f(\gamma(1),\gamma(0)|R,N) f(X_{\bullet,\sigma(1)}, \ldots, X_{\bullet,\sigma(N)}|R,N)f(R,N).
\end{align*}
By assumption, we have $f(X_{\bullet,\sigma(1)}, \ldots, X_{\bullet,\sigma(N)}|R,N)=f(X_{\bullet,1}, \ldots, X_{\bullet,N}|R,N)$, and under the bivariate random coefficient model \eqref{eq:bivariate}, we have $\prod_{j=1}^N f(Y_{\bullet,j}(1),Y_{\bullet,j}(0)|\gamma(1),\gamma(0),X_{\bullet,j},R,N)=\prod_{j=1}^N f(Y_{\bullet,\sigma(j)}(1),Y_{\bullet,\sigma(j)}(0)|\gamma(1),\gamma(0),X_{\bullet,\sigma(j)},R,N)$; thus we have $f(W_{\bullet,\sigma(1)}, \dots, W_{\bullet,\sigma(N)})=f(W_{\bullet,1}, \dots, W_{\bullet,N})$ and within-cluster exchangeability holds. 

\section{Proofs}\label{sec:proofs}
\subsection{Proof of Theorem 1}
\begin{proof}[Proof of Theorem 1.]
The proof of Theorem 1 follows the classical split conformal causal inference \citep{lei2021conformal}, where the key difference in our method is that we only consider the data within \(\Omega_C\) to obtain the local coverage.

Since $\widetilde{\textup{C}}_{C}(\overline{O}_{\textup{test}}) = (-1)^{A_{\textup{test}} + 1} \{\overline{Y}_{\textup{test}} - \widetilde{\textup{C}}_{C,1-A_{\textup{test}}}(\overline{B}_{\textup{test}})\}$ and $\overline{Y}_{\textup{test}} = A_{\textup{test}}\overline{Y}_{\textup{test}}(1)+(1-A_{\textup{test}})\overline{Y}_{\textup{test}}(0)$, we have
\begin{align*}
     &P\bigg\{\overline{Y}_{\textup{test}}(1) - \overline{Y}_{\textup{test}}(0) \not\in \widetilde{\textup{C}}_C(\overline{O}_{\textup{test}})\bigg|\overline{B}_{\textup{test}} \in \Omega_C\bigg\} \\
     &= \sum_{a=0}^1 P\bigg\{\overline{Y}_{\textup{test}}(1) - \overline{Y}_{\textup{test}}(0) \not\in \widetilde{\textup{C}}_C(\overline{O}_{\textup{test}}),  A_{\textup{test}}=a\bigg|\overline{B}_{\textup{test}} \in \Omega_C\bigg\} \\
     &=  \sum_{a=0}^1 P\bigg\{\overline{Y}_{\textup{test}}(a)\not\in \widetilde{\textup{C}}_{C,a}(\overline{B}_{\textup{test}}),  A_{\textup{test}}=1-a\bigg|\overline{B}_{\textup{test}} \in \Omega_C\bigg\} \\
     &\le \sum_{a=0}^1 P\bigg\{\overline{Y}_{\textup{test}}(a)\not\in \widetilde{\textup{C}}_{C,a}(\overline{B}_{\textup{test}})\bigg|\overline{B}_{\textup{test}} \in \Omega_C\bigg\}.
\end{align*}
Therefore, it remains to show $P\bigg\{\overline{Y}_{\textup{test}}(a)\not\in \widetilde{\textup{C}}_{C,a}(\overline{B}_{\textup{test}})\bigg|\overline{B}_{\textup{test}} \in \Omega_C\bigg\} \le \alpha$ for $a = 0,1$. 

Denote $\mathcal{I} = \mathcal{I}_{ca}(a) \cup \{\textup{test}\}$. By Assumptions 1-2 and the additional assumption in Theorem 1, $\{(\overline{Y}_i(a), \overline{B}_i): i \in \mathcal{I}\}$ are i.i.d. given $\overline{B}_i \in \Omega_C$. 
Of note, here we use Assumption 2 to obtain $(\overline{Y}_i(a), \overline{B}_i)|(A_i=a, \overline{B}_i \in \Omega_C)$ is identically distributed as $(\overline{Y}_i(a), \overline{B}_i)|(\overline{B}_i \in \Omega_C)$ for $i \in \mathcal{I}_{ca}(a)$.
Since $\widehat{f}_a$ is a function of the training data, which are independent of the calibration data and test data, then $\{s(\overline{B}_i,\overline{Y}_i(a)), i \in \mathcal{I}\}$ are also i.i.d. conditioning on the training fold and $\overline{B}_i \in \Omega_C$. Then,
\begin{align*}
    P\left\{s(\overline{B}_{\textup{test}},\overline{Y}_{\textup{test}}(a)) \le \textup{Quantile}_{1-\alpha}\left(\frac{1}{|\mathcal{I}| }\sum_{i \in \mathcal{I}} \delta_{s(\overline{B}_i,\overline{Y}_i(a))}\right)\bigg|\overline{B}_{\textup{test}} \in \Omega_C\right\} \ge 1-\alpha.
\end{align*}
Since the quantile of an empirical distribution does not decrease if we replace one sample by $+\infty$, we also have
\begin{align*}
    P\left\{s(\overline{B}_{\textup{test}},\overline{Y}_{\textup{test}}(a)) \le \textup{Quantile}_{1-\alpha}\left(\frac{1}{|\mathcal{I}_{ca}(a)|+1 }\sum_{i \in \mathcal{I}_{ca}(a)} \delta_{s(\overline{B}_i,\overline{Y}_i(a))} + \delta_{+\infty}\right)\bigg| \overline{B}_{\textup{test}} \in \Omega_C\right\} \ge 1-\alpha.
\end{align*}
Since $\overline{Y}_i(a) = \overline{Y}_i$ for $i \in \mathcal{I}_{ca}(a)$, we have
\begin{align*}
    \textup{Quantile}_{1-\alpha}\left(\frac{1}{|\mathcal{I}_{ca}(a)| +1}\sum_{i \in \mathcal{I}_{ca}(a)} \delta_{s(\overline{B}_i,\overline{Y}_i(a))} + \delta_{+\infty}\right) = \textup{Quantile}_{1-\alpha}(\widehat{F}) = \widehat{q}_{1-\alpha}(a).
\end{align*}
Therefore, 
\begin{align*}
    P\left(\overline{Y}_{\textup{test}}(a)\in \widetilde{\textup{C}}_{C,a}(\overline{B}_{\textup{test}})\Big| \overline{B}_{\textup{test}} \in \Omega_C\right) = P\left(s(\overline{B}_{\textup{test}},\overline{Y}_{\textup{test}}(a)) \le \widehat{q}_{1-\alpha}(a)\Big| \overline{B}_{\textup{test}}  \in \Omega_C\right) \ge 1-\alpha,
\end{align*}
which completes the proof.
\end{proof}

\subsection{Proof of Theorem 2}
\begin{lemma}\label{lemma1}
Denote $M_i = \sum_{j=1}^{N_i} I\{ B_{ij} \in \Omega_I\}$ and $1\le j_1\le \dots\le j_{M_i} \le N_i$ be the ordered indices such that $ B_{ij_k} \in \Omega_I$ for $k =1,\dots, M_i$. Letting $e_{ij_k}$ be a length-$N_i$ vector with the $j_k$-th entry 1 and the rest zero, we define a matrix $D_i = [e_{ij_1},\dots, e_{ij_{M_i}}] \in \mathbb{R}^{N_i \times M_i}$ and $\widetilde{W}_i = D_i^\top W_i$, where $W_i = (V_{i1}^\top, \dots, V_{iN_i}^\top)^\top \in \mathbb{R}^{N_i\times (p+2)}$ for $V_{ij} = (Y_{ij}(1), Y_{ij}(0), B_{ij}) \in \mathbb{R}^{p+2}$. In other words, $\widetilde{W}_i$ is the data from cluster $i$ after filtering out individuals with $B_{ij} \not\in \Omega_I$.
For notational convenience, we denote the $j$-the row of $\widetilde{W}_i$ as $\widetilde{V}_{ij}$.
Under Assumptions 1 and 3, we have

(a) $(\widetilde{W}_1,\dots, \widetilde{W}_m)$ are independent and identically distributed.

(b) Within $\widetilde{W}_i$, the distribution of $(\widetilde{V}_{i1},\dots, \widetilde{V}_{iM_i})$ is exchangeable conditioning on $M_i$.

(c) Given the above results, $\widetilde{V}_{ij}$ has the same distribution as  $V_{ij}|(B_{ij} \in \Omega_I)$ for any $j$. It implies that $E[f(\widetilde{V}_{ij}, U)] = E[f(V_{ij}, U)|B_{ij} \in \Omega_I]$ for any $U$ that is independent of $(V_{ij}, \widetilde{V}_{ij})$ and any intergrable function $f$. 
\end{lemma}

\begin{proof}[Proof of Lemma~\ref{lemma1}.]
To prove (a), we observe that $D_i$ is a deterministic function of $W_i$ and $\Omega_I$. Since  $\widetilde{W}_i = D_i^\top W_i$, Assumption 1 implies that $(\widetilde{W}_1,\dots, \widetilde{W}_m)$ are independent and identically distributed.

To prove (b), conditional on $N_i$, we define $Q_i=(I\{B_{i1} \in \Omega_I\}, \dots, I\{B_{iN_i} \in \Omega_I\})$. We first prove that $(\widetilde{V}_{i1},\dots, \widetilde{V}_{iM_i})| Q_i=q, M_i=m$ are exchangeable.

Without loss of generality, we assume  $B_{ij},j=1,\cdots,m\in \Omega_I$ (the first $m$ elements). For any measureable sets $A_1,\cdots,A_m$, the conditional distribution of $(\widetilde{V}_{i1},\dots, \widetilde{V}_{iM_i})| Q_i=q, M_i=m$ is given by
\begin{align*}
(i):=P(V_{i1}\in A_1,\cdots, V_{im}\in A_m| V_{i1}\in B,\cdots, V_{im}\in B, V_{i(m+1)}\in C, \cdots, V_{i N_i}\in C).
\end{align*}
Here we denote the domain of $(Y(1),Y(0))$ as $\mathcal{D}_Y\in \mathbb{R}^2$ and let $B=\mathcal{D}_Y\times \Omega_I$ and $C=\mathcal{D}_Y\times \Omega_I^c.$
It holds that
$$(i)=\frac{P(V_{i1}\in A_1\cap B,\cdots, V_{im}\in A_m\cap B, V_{i(m+1)}\in C, \cdots, V_{i N_i}\in C)}{P(V_{i1}\in B,\cdots, V_{im}\in B, V_{i(m+1)}\in C, \cdots, V_{i N_i}\in C)} $$
Since conditional on $N_i$, $(V_{i1},\cdots,V_{iN_i})$ are exchangeable, for any permutation map $\sigma\in S_m$ (where $S_m$ is the permutation group on $1,\dots,m$), we obtain
\begin{align*}
 (i)&=\frac{P(V_{i\sigma(1)}\in A_1\cap B,\cdots, V_{i\sigma(m)}\in A_m\cap B, V_{i(m+1)}\in C, \cdots, V_{i N_i}\in C)}{P(V_{i1}\in B,\cdots, V_{im}\in B, V_{i(m+1)}\in C)}\\&=\frac{P(V_{i\sigma(1)}\in A_1,\cdots, V_{i\sigma(m)}\in A_m, V_{i1}\in B,\cdots, V_{im}\in B,V_{i(m+1)}\in C, \cdots, V_{i N_i}\in C)}{P(V_{i1}\in B,\cdots, V_{im}\in B, V_{i(m+1)}\in C) }\\&=P(V_{i\sigma(1)}\in A_1,\cdots, V_{i\sigma(m)}\in A_m|  V_{i1}\in B,\cdots, V_{im}\in B,V_{i(m+1)}\in C, \cdots, V_{i N_i}\in C).
 \end{align*}
Therefore, we have $[(\widetilde{V}_{i1},\dots, \widetilde{V}_{iM_i})| Q_i=q, M_i=m] \overset{d}{=}[(\widetilde{V}_{i\sigma(1)},\dots, \widetilde{V}_{i\sigma(M_i)})| Q_i=q, M_i=m]$ for any $\sigma\in S_m.$

We next prove the exchangeability by only conditioning on $M_i=m.$

We know that for any $\sigma\in S_m$
\begin{align*}
&P(\widetilde{V}_{i\sigma(1)},\dots, \widetilde{V}_{i\sigma(M_i)}\in A_1\times\cdots \times A_m| M_i=m)\\&\qquad=\sum_{Q_i=q}P(\widetilde{V}_{i\sigma(1)},\dots, \widetilde{V}_{i\sigma(M_i)}\in A_1\times\cdots\times A_m| Q_i=q,M_i=m)\cdot P(Q_i=q| M_i=m)
\\&\qquad =\sum_{Q_i=q}P(\widetilde{V}_{i1},\dots, \widetilde{V}_{iM_i}\in A_1\times\cdots\times A_m| Q_i=q,M_i=m)\cdot P(Q_i=q| M_i=m)
\\&\qquad =P(\widetilde{V}_{i1},\dots, \widetilde{V}_{iM_i}\in A_1\times\cdots\times A_m | M_i=m).
\end{align*}
We then conclude our proof of (b).

Finally, we prove (c). By definition, $\widetilde{V}_{ij}$ is equal to $V_{ij'}$ for some $j'$ such that $B_{ij'} \in \Omega_I$. Therefore, $\widetilde{V}_{ij} \overset{d}{=} V_{ij'}|(B_{ij'} \in \Omega_I)$. By the exchangeability of $W_i$, we have  $\widetilde{V}_{ij} \overset{d}{=} V_{ij}|(B_{ij} \in \Omega_I)$. In addition, 
\begin{align*}
    E[f(V_{ij}, U)|B_{ij} \in \Omega_I] &= \int f(V_{ij}, U) dP(V_{ij}, U|B_{ij} \in \Omega_I) \\
    &= \int f(V_{ij}, U) d\{P(V_{ij}|B_{ij} \in \Omega_I)P(U)\}\\
    &= \int f(\widetilde{V}_{ij}, U) d\{P(\widetilde{V}_{ij})P(U)\}\\
    &= \int f(\widetilde{V}_{ij}, U) d P(\widetilde{V}_{ij},U)\\
    &= E[f(\widetilde{V}_{ij}, U)],
\end{align*}
where the second and fourth equation uses the independence between $U$ and $(V_{ij},\widetilde{V}_{ij})$, and the thrid equation results from $\widetilde{V}_{ij} \overset{d}{=} V_{ij}|(B_{ij} \in \Omega_I)$.
\end{proof}

\begin{proof}[Proof of Theorem 2.]
Since $\widetilde{\textup{C}}_{I}(O_{\textup{test}}) = (-1)^{A_{\textup{test}} + 1} \{Y_{\textup{test}} - \widetilde{\textup{C}}_{I,1-A_{\textup{test}}}(B_{\textup{test}})\}$ and $Y_{\textup{test}} = A_{\textup{test}}Y_{\textup{test}}(1)+(1-A_{\textup{test}})Y_{\textup{test}}(0)$, we have
\begin{align*}
     &P\bigg\{Y_{\textup{test}}(1) - Y_{\textup{test}}(0) \not\in \widetilde{\textup{C}}_I(O_{\textup{test}})\bigg|B_{\textup{test}} \in \Omega_C\bigg\} \\
     &= \sum_{a=0}^1 P\bigg\{Y_{\textup{test}}(1) - Y_{\textup{test}}(0) \not\in \widetilde{\textup{C}}_I(O_{\textup{test}}),  A_{\textup{test}}=a\bigg|B_{\textup{test}} \in \Omega_C\bigg\} \\
     &=  \sum_{a=0}^1 P\bigg\{Y_{\textup{test}}(a)\not\in \widetilde{\textup{C}}_{I,a}(B_{\textup{test}}),  A_{\textup{test}}=1-a\bigg|B_{\textup{test}} \in \Omega_C\bigg\} \\
     &\le \sum_{a=0}^1 P\bigg\{Y_{\textup{test}}(a)\not\in \widetilde{\textup{C}}_{I,a}(B_{\textup{test}})\bigg|B_{\textup{test}} \in \Omega_C\bigg\}.
\end{align*}
Therefore, it remains to show $P\bigg\{Y_{\textup{test}}(a)\not\in \widetilde{\textup{C}}_{I,a}(B_{\textup{test}})\bigg|B_{\textup{test}} \in \Omega_C\bigg\} \le \alpha$ for $a = 0,1$.

To this end, we denote $W_{m+1}$ as a new cluster independently sampled from $\mathcal{P}^{W}$. By construction, $(Y_{\textup{test}}(1), Y_{\textup{test}}(0), B_{\textup{test}})$ comes from an arbitrary individual in this cluster. By Assumption 3, $(Y_{\textup{test}}(1), Y_{\textup{test}}(0), B_{\textup{test}})$ is identically distributed as $V_{m+1,1}=(Y_{m+1,1}(1), Y_{m+1,1}(0), B_{m+1,1})$, where $V_{m+1,j}$ is the $j$-th row of $W_{m+1}$. Therefore, our goal is to show $P\bigg\{Y_{m+1,1}(a)\not\in \widetilde{\textup{C}}_{I,a}(B_{m+1,1})\bigg|B_{m+1,1} \in \Omega_C\bigg\} \le \alpha$ for $a = 0,1$. To further simplify the goal, we denote $\widetilde{W}_{m+1}$ as the new cluster after we filter out individuals with $B_{m+1,j} \not\in \Omega_I$ and let $\widetilde{V}_{m+1,j} = (\widetilde{Y}_{m+1,j}(1), \widetilde{Y}_{m+1,j}(0), \widetilde{B}_{m+1,j})$ denote the $j$-th row of $\widetilde{W}_{m+1}$. (See Lemma~\ref{lemma1} for the detailed construction for $\widetilde{W}_{m+1}$). By Lemma~\ref{lemma1} (c), $\widetilde{V}_{m+1,1}$ is identically distributed as $V_{m+1,1}|(B_{m+1,1} \in \Omega_I)$. Since $\widetilde{\textup{C}}_{I,a}$ is a function of $\{W_i: i=1,\dots, m, A_i=a\}$, which are independent of $W_{m+1}$, Lemma~\ref{lemma1}(c) further implies
\begin{align*}
    P\bigg\{Y_{m+1,1}(a)\not\in \widetilde{\textup{C}}_{I,a}(B_{m+1,1})\bigg|B_{m+1,1} \in \Omega_C\bigg\} = P\left\{\widetilde{Y}_{m+1,1}(a)\not\in \widetilde{\textup{C}}_{I,a}(\widetilde{B}_{m+1,1})\right\}.
\end{align*}
This result yields our final goal to show $P\left\{\widetilde{Y}_{m+1,1}(a)\not\in \widetilde{\textup{C}}_{I,a}(\widetilde{B}_{m+1,1})\right\} \le \alpha$.

Next, we denote the calibration data $\mathcal{O}_{ca}(a) = \{\widetilde{W}_i: i \in \mathcal{I}_{ca}(a)\}$. Like $\widetilde{W}_{m+1}$, $\widetilde{W}_i$ contains those individuals in $W_i$ with $B_{ij} \in \Omega_I$. 
By Assumptions 1-2 and Lemma~\ref{lemma1}(a), $\{\widetilde{W}_i: \in \mathcal{I}\}$ are i.i.d., where $ \mathcal{I} =  \mathcal{I}_{ca}(a) \cup \{m+1\}$. Within each $\widetilde{W}_i$, Lemma~\ref{lemma1} (b) implies that $(\widetilde{V}_{i1}, \dots, \widetilde{V}_{iM_i})$ are exchangeable conditioning on $M_i$ within each cluster. 

We define a function 
\begin{align*}
    q_{1-\alpha}(\{\widetilde{W}_i, i\in \mathcal{I}\}) = \textup{Quantile}_{1-\alpha}\left(\frac{1}{|\mathcal{I}|}\sum_{i\in \mathcal{I}} \frac{1}{M_i}\sum_{j=1}^{M_i} \delta_{s(\widetilde{B}_{ij}, \widetilde{Y}_{ij}(a))}\right),
\end{align*}
which is the $(1-\alpha)$-quantile of the weighted empirical distribution for the calibration data and new cluster. By definition of the quantile function, we have, for any $i \in \mathcal{I}$ and $j \in \{1,\dots, M_i\}$,
\begin{align*}
    \frac{1}{|\mathcal{I}|}\sum_{i\in \mathcal{I}} \frac{1}{M_i}\sum_{j=1}^{M_i} I\left\{s(\widetilde{B}_{ij}, \widetilde{Y}_{ij}(a)) \le  q_{1-\alpha}(\{\widetilde{W}_i, i\in \mathcal{I}\})\right\} \ge 1-\alpha.
\end{align*}
{Consider any permutation $\sigma_i$ on $(1,\dots, M_{i})$}. Since we have shown that each $\widetilde{W}_i$ is exchangeable given $M_i$, we have, for each $i \in \mathcal{I}$,
\begin{align*}
    I\left\{s(\widetilde{B}_{ij}, \widetilde{Y}_{ij}(a)) \le  q_{1-\alpha}(\{\widetilde{W}_i, i\in \mathcal{I}\})\right\} \overset{d}{=} I\left\{s(\widetilde{B}_{i\sigma(j)}, \widetilde{Y}_{i\sigma(j)}(a)) \le  q_{1-\alpha}(\{\widetilde{W}_{i'}, i'\in \mathcal{I}, i'\ne i\} \cup \{\sigma(\widetilde{W}_i)\})\right\},
\end{align*}
conditioning on the training fold and $M_i$. 
Since the weighted empirical distribution is invariant to within-cluster permutations, we have $q_{1-\alpha}(\{\widetilde{W}_{i'}, i'\in \mathcal{I}, i'\ne i\} \cup \{\sigma(\widetilde{W}_i)\}) = q_{1-\alpha}(\{\widetilde{W}_i, \in \mathcal{I}\})$.
Combined with the fact that the training fold is independent of the calibration data and the test data, we have
\begin{align*}
    E\left[I\left\{s(\widetilde{B}_{ij}, \widetilde{Y}_{ij}(a)) \le  q_{1-\alpha}(\{\widetilde{W}_i, \in \mathcal{I}\})\right\}\Big|M_i\right] =  E\left[I\left\{s(\widetilde{B}_{i1}, \widetilde{Y}_{i1}(a)) \le  q_{1-\alpha}(\{\widetilde{W}_i, \in \mathcal{I}\})\right\}\Big|M_i\right] 
\end{align*}
by choosing permutations $\sigma$ with $\sigma(j) = 1$. After averaging over $i$ and marginalizing over $M_i$, we have
\begin{align*}
   E\left[I\left\{s(\widetilde{B}_{i1}, \widetilde{Y}_{i1}(a)) \le  q_{1-\alpha}(\{\widetilde{W}_i, \in \mathcal{I}\})\right\}\right] =  E\left[\frac{1}{M_i}\sum_{j=1}^{M_i}I\left\{s(\widetilde{B}_{ij}, \widetilde{Y}_{ij}(a)) \le  q_{1-\alpha}(\{\widetilde{W}_i, \in \mathcal{I}\})\right\}\right].
\end{align*}
Since we have shown that $\widetilde{W}_i$ are i.i.d., we also have
\begin{align*}
    & E\left[\frac{1}{M_i}\sum_{j=1}^{M_i}I\left\{s(\widetilde{B}_{ij}, \widetilde{Y}_{ij}(a)) \le  q_{1-\alpha}(\{\widetilde{W}_i, \in \mathcal{I}\})\right\}\right]\\ &= E\left[\frac{1}{|\mathcal{I}|} \sum_{i\in \mathcal{I}}\frac{1}{M_i}\sum_{j=1}^{M_i}I\left\{s(\widetilde{B}_{ij}, \widetilde{Y}_{ij}(a)) \le  q_{1-\alpha}(\{\widetilde{W}_i, \in \mathcal{I}\})\right\}\right].
\end{align*}
Taken together, we get
\begin{align*}
    & E\left[I\left\{s(\widetilde{B}_{m+1,1}, \widetilde{Y}_{m+1,1}(a)) \le  q_{1-\alpha}(\{\widetilde{W}_i, \in \mathcal{I}\})\right\}\right]\\ &= E\left[\frac{1}{|\mathcal{I}|} \sum_{i\in \mathcal{I}}\frac{1}{M_i}\sum_{j=1}^{M_i}I\left\{s(\widetilde{B}_{ij}, \widetilde{Y}_{ij}(a)) \le  q_{1-\alpha}(\{\widetilde{W}_i, \in \mathcal{I}\})\right\}\right] \\
    &\ge 1-\alpha.
\end{align*}
Finally, we need to connect $q_{1-\alpha}(\{\widetilde{W}_i, \in \mathcal{I}\})$ to $\widehat{q}_{1-\alpha}$ defined in Algorithm 3. Since replacing some point masses in the weighted empirical distribution by $\delta_{+\infty}$ does not decrease the $(1-\alpha)$-quantile, we have $q_{1-\alpha}(\{\widetilde{W}_i, \in \mathcal{I}\})\le\widehat{q}_{1-\alpha}$, which implies $P(s(\widetilde{B}_{m+1,1}, \widetilde{Y}_{m+1,1}(a)) \le \widehat{q}_{1-\alpha}) \ge 1-\alpha$. Since we defined $\widetilde{\textup{C}}_{I,a}(B) = \{y \in \mathbb{R}: |y-\widehat{f}_a(B)| \le \widehat{q}_{1-\alpha}(a)\}$, then the event $Y_{m+1,1} \in\widetilde{\textup{C}}_{I,a}(B_{m+1,1})$  is equal to $s(\widetilde{B}_{m+1,1}, \widetilde{Y}_{m+1,1}(a)) \le \widehat{q}_{1-\alpha}$, which implies $P\left\{\widetilde{Y}_{m+1,1}(a)\not\in \widetilde{\textup{C}}_{I,a}(\widetilde{B}_{m+1,1})\right\} \le \alpha$. This completes the proof.
\end{proof}

\section{Nested approaches to construct conformal intervals}\label{sec: nested}
\subsection{Cluster-level treatment effect}
\begin{algorithm}[H]
\setstretch{1}
\caption{\label{alg:2} Computing the conformal interval $\widetilde{\textup{C}}_C(\overline{B})$ for cluster-level treatment effects.}
\textbf{Input:} Cluster-level data $\{(\overline{Y}_i, A_i, \overline{B}_i):i=1,\dots, m\}$,  a test point $\overline{B}_{\textup{test}}$, a prediction model $f_a(\overline{B})$ for $\overline{Y}(a)$, $a\in \{0,1\},$ prediction models $\{m^L(\overline{B}), m^R(\overline{B})\}$ for the $(\alpha/2, 1-\alpha/2)$-quantiles of $\overline{Y}(1) -\overline{Y}(0)$,
a covariate-subgroup of interest $\Omega_C$, and levels $(\alpha,\gamma)$.

\vspace{5pt}
\textbf{Procedure:}  

\hspace{5pt} 1. For $a = 0,1$, randomly split the arm$-a$ covariate-subgroup data $\{(\overline{Y}_i, \overline{B}_i): i=1,\dots, m, A_i =a , \overline{B}_i \in \Omega_C\}$ into a training fold $\mathcal{O}_{tr}(a)$, and a calibration fold $\mathcal{O}_{ca}(a)$ with index set $\mathcal{I}_{ca}(a)$.

\hspace{5pt} 2. Use the training fold to run Step 1 of Algorithm~1 of the main paper and obtain the 
$(1-\alpha)$ conformal interval $\widetilde{\textup{C}}_{C,a}(\overline{B})$ for $\overline{Y}(a), a = 0,1$.

\hspace{5pt} 3. For each $i =1,\dots, m$, define
    $\overline{C}_i = A_i\{\overline{Y_i} - \widetilde{\textup{C}}_{C,0}(\overline{B}_i)\} + (1-A_i) \{\widetilde{\textup{C}}_{C,1}(\overline{B}_i) - \overline{Y_i}\}$
and denote $\overline{C}_i = [\overline{C}_i^L, \overline{C}_i^R]$. 




\hspace{5pt} 4. Train prediction models $m^L(\overline{B})$ for $\overline{C}^L$ and $m^R(\overline{B})$  for $\overline{C}^R$ using all data in the training fold $\{\mathcal{O}_{tr}(a): a=0,1\}$, and obtain the estimated models $\widehat{m}^L(\overline{B})$ and $\widehat{m}^R(\overline{B})$.

\hspace{5pt} 5. For each $i \in \mathcal{I}_{ca}(1) \cup \mathcal{I}_{ca}(0)$, compute the non-conformity score $$s^*(\overline{B}_i, \overline{C}_i) = \max\{\widehat{m}^L(\overline{B}_i) - \overline{C}_i^L,  \overline{C}_i^R - \widehat{m}^R(\overline{B}_i)\}.$$

\hspace{5pt} 6. Compute the $1-\gamma$ quantile $\widehat{q}^*_{1-\gamma}$ of the distribution $$\widehat{F}^*= \frac{1}{|\mathcal{I}_{ca}(1)|+|\mathcal{I}_{ca}(0)|+1}\left\{\sum_{i \in \mathcal{I}_{ca}(1)\cup \mathcal{I}_{ca}(0)}   \delta_{s^*(\overline{B}_i, \overline{C}_i)} + \delta_{+\infty}\right\}.$$
\vspace{5pt}
\textbf{Output:} $\widetilde{\textup{C}}_C(\overline{B}_{\textup{test}}) = [\widehat{m}^L(\overline{B}_{\textup{test}}) - \widehat{q}^*_{1-\gamma}, \widehat{m}^R(\overline{B}_{\textup{test}}) + \widehat{q}^*_{1-\gamma}]$.
\end{algorithm}

\begin{theorem}\label{prop:2}
Assume Assumptions 1-2 and that $(\overline{Y}_{\textup{test}}(1),\overline{Y}_{\textup{test}}(0),\overline{B}_{\textup{test}})$ is an independent from the distribution induced by $\mathcal{P}^W$. Then, the $\widetilde{\textup{C}}_{C}(\overline{B}_{\textup{test}})$ output by Algorithm~\ref{alg:2} satisfies
\begin{equation}\label{eq: prop2}
    P\bigg\{\overline{Y}_{\textup{test}}(1) - \overline{Y}_{\textup{test}}(0) \in \widetilde{\textup{C}}_C(\overline{B}_{\textup{test}})\bigg|\overline{B}_{\textup{test}} \in \Omega_C\bigg\} \ge 1-\alpha-\gamma
\end{equation}
for any set $\Omega_C$ in the support of $\overline{B}_{\textup{test}}$ with a positive measure. 
\end{theorem}

\begin{proof}[Proof of Theorem~\ref{prop:2}.]
Following the proof of Theorem 1, we have $P\bigg\{\overline{Y}(a)\not\in \widetilde{\textup{C}}_{C,a}(\overline{B})\bigg|\overline{B} \in \Omega_C\bigg\} \le \alpha$ for $a = 0,1$. Since Assumption 2 implies that $A$ is independent of $(\overline{Y}(a), \overline{B})$, then
\begin{align*}
    &P\left\{\overline{Y}(1) - \overline{Y}(0) \in  A\{\overline{Y} - \widetilde{\textup{C}}_{C,0}(\overline{B})\} + (1-A) \{\widetilde{\textup{C}}_{C,1}(\overline{B}) - \overline{Y}\}\Big|\overline{B} \in \Omega_C\right\} \\
    &= \sum_{a=0}^1P\left\{\overline{Y}(a) \in  \widetilde{\textup{C}}_{C,a}(\overline{B}), A = a\Big| \overline{B} \in \Omega_C\right\} \\
    &= \sum_{a=0}^1P\left\{\overline{Y}(a) \in  \widetilde{\textup{C}}_{C,a}(\overline{B})\Big| \overline{B} \in \Omega_C\right\} P(A=a)\\
    &\ge 1-\alpha.
\end{align*}
Therefore, by defining $\overline{C}_i = A_i\{\overline{Y_i} - \widetilde{\textup{C}}_{C,0}(\overline{B}_i)\} + (1-A_i) \{\widetilde{\textup{C}}_{C,1}(\overline{B}_i) - \overline{Y_i}\}$, we have $\overline{C}_i = [\overline{C}_i^L, \overline{C}_i^R]$ representing the $(\alpha/2,1-\alpha/2)$-quantile of $\overline{Y}_i(1) - \overline{Y}_i(0)$. 

Next, we independently generate $A_{\textup{test}}$  from $\mathcal{P}^A$ and compute $\overline{C}_{\textup{test}} = A_{\textup{test}}\{\overline{Y}_{\textup{test}} - \widetilde{\textup{C}}_{C,0}(\overline{B}_{\textup{test}})\} + (1-A_{\textup{test}}) \{\widetilde{\textup{C}}_{C,1}(\overline{B}_{\textup{test}}) - \overline{Y}_{\textup{test}}\}$. As a result, $(\overline{B}_{\textup{test}}, \overline{C}_{\textup{test}})$ is independent and identically distributed as $(\overline{B}_i, \overline{C}_i)$.
We observe that $\widetilde{\textup{C}}_{C,a}, \widehat{m}^L, \widehat{m}^R$ are all functions of the training folds $\{\mathcal{O}_{tr}(a): a=0,1\}$, which are independent of the calibration data and test data. Denoting $\mathcal{I} = \mathcal{I}_{ca}(1) \cup \mathcal{I}_{ca}(0) \cup \{\textup{test}\}$, by Assumptions 1-2, $\{s^*(\overline{B}_i,\overline{C}_i), i \in \mathcal{I}\}$ are i.i.d. conditioning on the training folds and $\overline{B}_i \in \Omega_C$. Then, we have
\begin{align*}
    P\left\{s^*(\overline{B}_{\textup{test}},\overline{C}_{\textup{test}}) \le \textup{Quantile}_{1-\alpha}\left(\frac{1}{|\mathcal{I}| }\sum_{i\in \mathcal{I}_{ca}(1)\cup \mathcal{I}_{ca}(0)} \delta_{s^*(\overline{B}_i,\overline{C}_i)} + \delta_{+\infty}\right)\bigg| \overline{B} \in \Omega_C\right\} \ge 1-\gamma.
\end{align*}
Since we defined
\begin{align*}
    \textup{Quantile}_{1-\gamma}\left(\frac{1}{|\mathcal{I}| }\sum_{i \in \mathcal{I}_{ca}(a)} \delta_{s^*(\overline{B}_i,\overline{C}_i)} + \delta_{+\infty}\right) =  \textup{Quantile}_{1-\gamma}(\widehat{F}^*) = \widehat{q}^*_{1-\gamma},
\end{align*}
we obtain 
\begin{align*}
 P\{\overline{C}_{\textup{test}} \subset \widetilde{\textup{C}}_C(\overline{B}_{\textup{test}})|\overline{B} \in \Omega_C\} &= P\bigg\{\widehat{m}^L(\overline{B}_{\textup{test}})-\widehat{q}^*_{1-\gamma} \le C_{\textup{test}}^L, C_{\textup{test}}^R \le \widehat{m}^R(\overline{B}_{\textup{test}})+\widehat{q}^*_{1-\gamma}\bigg|\overline{B} \in \Omega_C\bigg\} \\
 &= P\{\max\{\widehat{m}^L(\overline{B}_{\textup{test}})-C_{\textup{test}}^L, C_{\textup{test}}^R - \widehat{m}^R(\overline{B}_{\textup{test}})\} \le \widehat{q}^*_{1-\gamma}|\overline{B} \in \Omega_C\} \\
 &= P(s^*(\overline{B}_{\textup{test}},\overline{C}_{\textup{test}}) \le \widehat{q}^*_{1-\gamma}|\overline{B} \in \Omega_C)\\
 &\ge 1-\gamma. 
\end{align*}
Combined with the fact that $P\{\overline{Y}_{\textup{test}}(1)-\overline{Y}_{\textup{test}}(0)\in C_{\textup{test}}|\overline{B} \in \Omega_C\} \ge 1-\alpha$, we have
\begin{align*}
    &P\{\overline{Y}_{\textup{test}}(1)-\overline{Y}_{\textup{test}}(0)\not\in \widetilde{\textup{C}}_C(\overline{B}_{\textup{test}})|\overline{B} \in \Omega_C\} \\
    &=P\{\overline{Y}_{\textup{test}}(1)-\overline{Y}_{\textup{test}}(0)\not\in \widetilde{\textup{C}}_C(\overline{B}_{\textup{test}}), \overline{C}_{\textup{test}} \subset \widetilde{\textup{C}}_C(\overline{B}_{\textup{test}})|\overline{B} \in \Omega_C\} \\
    &\quad + P\{\overline{Y}_{\textup{test}}(1)-\overline{Y}_{\textup{test}}(0)\not\in \widetilde{\textup{C}}_C(\overline{B}_{\textup{test}}), \overline{C}_{\textup{test}} \not\subset \widetilde{\textup{C}}_C(\overline{B}_{\textup{test}})|\overline{B} \in \Omega_C\} \\
    &\le P\{\overline{Y}_{\textup{test}}(1)-\overline{Y}_{\textup{test}}(0)\not\in \overline{C}_{\textup{test}}|\overline{B} \in \Omega_C\} + P\{\overline{C}_{\textup{test}} \not\subset \widetilde{\textup{C}}_C(\overline{B}_{\textup{test}})|\overline{B} \in \Omega_C\} \\
    &\le \alpha+\gamma,
\end{align*}
which completes the proof.

\end{proof}

\subsection{Individual-level treatment effect}
\begin{algorithm}[htb]
\caption{\label{alg:4} Computing the conformal interval $\widetilde{\textup{C}}_I(B)$ for indivdiual-level treatment effects.}
\textbf{Input:} 
Individual-level data $\{(Y_{ij}, A_i, B_{ij}):i=1,\dots,m; j = 1,\dots, N_i\}$,  a test point $B_{\textup{test}}$, a prediction model $f_a(B)$ for $Y(a)$, $a\in \{0,1\},$ prediction models $\{m^L(B), m^R(B)\}$ for the $(\alpha/2, 1-\alpha/2)$-quantiles of $Y(1) -Y(0)$, a covariate-subgroup of interest $\Omega_I$, and levels $(\alpha, \gamma)$.

\vspace{5pt}
\textbf{Procedure:}  

\hspace{5pt} 1. For $a =0,1$, randomly partition the arm$-a$ covariate-subgroup data $\{(Y_{ij}, A_i, B_{ij}):i=1,\dots,m; j = 1,\dots, N_i; A_i=a; B_{ij} \in \Omega_I\}$ into a training fold $\mathcal{O}_{tr}(a)$ and a calibration fold $\mathcal{O}_{ca}(a)$ with index set $\mathcal{I}_{ca}(a)$.

\hspace{5pt} 2. Use the training fold to run Step 1 of Algorithm~2 of the main paper and obtain the $(1-\alpha)$ conformal interval $\widetilde{\textup{C}}_{I,a}(B)$ for $Y(a), a = 0,1$.

\hspace{5pt} 3. For each $i =1,\dots, m$, define $C_{ij} = A_i\{Y_{ij} - \widetilde{\textup{C}}_{I,0}(B_{ij})\} + (1-A_i) \{\widetilde{\textup{C}}_{I,1}(B_{ij}) - Y_{ij}\}$
and denote $C_{ij} = [C_{ij}^L, C_{ij}^R]$. 




\hspace{5pt} 4. Train prediction models $m^L(B)$ for $C^L$ and $m^R(B)$  for $C^R$ using all data in the training fold $\{\mathcal{O}_{tr}(a): a=0,1\}$, and obtain the estimated models $\widehat{m}^L(B)$ and $\widehat{m}^R(B)$.

\hspace{5pt} 5. For each $(i,j) \in \mathcal{I}_{ca}(1) \cup \mathcal{I}_{ca}(0)$, compute the non-conformity score $$s^*(B_{ij}, C_{ij}) = \max\{\widehat{m}^L(B_{ij}) - C_{ij}^L,  C_{ij}^R - \widehat{m}^R(B_{ij})\}.$$

\hspace{5pt} 6. Compute the $1-\gamma$ quantile $\widehat{q}^*_{1-\gamma}$ of the distribution $$\widehat{F}^*= \frac{1}{|\mathcal{I}_{ca}(1)|+|\mathcal{I}_{ca}(0)|+1}\left\{ \sum_{i \in \mathcal{I}_{ca}(1)\cup \mathcal{I}_{ca}(0)}  \frac{1}{\sum_{j=1}^{N_i} I\{B_{ij} \in \Omega_I\}}\sum_{j=1}^{N_i} I\{B_{ij} \in \Omega_I\} \delta_{s^*(B_{ij}, C_{ij})} + \delta_{+\infty}\right\}.$$
\vspace{5pt}
\textbf{Output:} $\widetilde{\textup{C}}_C(B_{\textup{test}}) = [\widehat{m}^L(B_{\textup{test}}) - \widehat{q}^*_{1-\gamma}, \widehat{m}^R(B_{\textup{test}}) + \widehat{q}^*_{1-\gamma}]$.
\end{algorithm}

\begin{theorem}\label{thm2}
Assume Assumptions 1-3 and that $(Y_{\textup{test}}(1), Y_{\textup{test}}(0), B_{\textup{test}})$ is an arbitrary individual from a new cluster independently sampled from $\mathcal{P}^W$. Then, the $\widetilde{\textup{C}}_{I}(B_{\textup{test}})$ output by Algorithm~\ref{alg:4} satisfies
\begin{equation}\label{eq: thm2}
    P\bigg\{Y_{\textup{test}}(1) - Y_{\textup{test}}(0) \in \widetilde{\textup{C}}_I(B_{\textup{test}})\bigg|B_{\textup{test}} \in \Omega_I\bigg\} \ge 1-\alpha-\gamma
\end{equation}
for any set $\Omega_I$ in the support of $B_{\textup{test}}$ with a positive measure. 
\end{theorem}

\begin{proof}[Proof of Theorem~\ref{thm2}.]
    Following the proof of Theorem 2, we have $P(\widetilde{Y}(a) \in \widetilde{\textup{C}}_{I,a}(\widetilde{B})) \ge 1-\alpha$, where $\widetilde{Y}(a)$ and $\widetilde{B}$ represent the potential outcome $Y(a)$ and covariates $B$ given $B \in \Omega_I$. Since Assumption 2 implies that $A$ is independent of $(\widetilde{Y}(a), \widetilde{B}))$, then
    \begin{align*}
    &P\left\{\widetilde{Y}(1) - \widetilde{Y}(0) \in  A\{Y - \widetilde{\textup{C}}_{C,0}(\widetilde{B})\} + (1-A) \{\widetilde{\textup{C}}_{C,1}(\widetilde{B}) - \widetilde{Y}\}\right\} \\
    &= \sum_{a=0}^1P\left\{\widetilde{Y}(a) \in  \widetilde{\textup{C}}_{C,a}(\widetilde{B}), A = a\right\} \\
    &= \sum_{a=0}^1P\left\{\widetilde{Y}(a) \in  \widetilde{\textup{C}}_{C,a}(\widetilde{B})\right\} P(A=a)\\
    &\ge 1-\alpha.
\end{align*}
Therefore, by defining $\widetilde{C}_{ij} = A_i\{\widetilde{Y}_{ij} - \widetilde{\textup{C}}_{C,0}(\widetilde{B}_{ij})\} + (1-A_i) \{\widetilde{\textup{C}}_{C,1}(\widetilde{B}_{ij}) - \widetilde{Y}_{ij}\}$, we have $\widetilde{C}_{ij} = [\widetilde{C}_{ij}^L, \widetilde{C}_{ij}^R]$ representing the $(\alpha/2,1-\alpha/2)$-quantile of $\widetilde{Y}_{ij}(1) - \widetilde{Y}_{ij}(0)$. 

For the new cluster, in addition to sampling potential outcomes and covariates from $\mathcal{P}^{W}$, we independently generate $A_{m+1}$ from $\mathcal{P}^A$, and compute $\widetilde{C}_{m+1,j} = A_{m+1} \{Y_{m+1,j} - \widetilde{\textup{C}}_{C,0}(B_{m+1,j})\} + (1-A_{m+1})\{\widetilde{\textup{C}}_{C,1}(\widetilde{B}_{m+1,j}) - \widetilde{Y}_{m+1,j}\}$. By construction, $(\widetilde{B}_{m+1,j}, \widetilde{C}_{m+1,j})$ is independently and identically distributed as $(\widetilde{B}_{ij}, \widetilde{C}_{ij})$. We observe that $\widetilde{\textup{C}}_{C,a}, \widehat{m}^L, \widehat{m}^R$ are all functions of the  training folds $\{\mathcal{O}_{tr}(a): a=0,1\}$, which are independent of the calibration data and test data. Denoting $\mathcal{I} = \mathcal{I}_{ca}(1) \cup \mathcal{I}_{ca}(0) \cup \{\textup{test}\}$, by Assumptions 1-2, $\{s^*(\widetilde{B}_i,\widetilde{C}_i), i \in \mathcal{I}\}$ are i.i.d. conditioning on the training folds and $\widetilde{B}_i \in \Omega_C$. Then, we have
\begin{align*}
    P\left\{s^*(\widetilde{B}_{m+1,1}, \widetilde{C}_{m+1,1}) \le \textup{Quantile}_{1-\gamma}\left(\frac{1}{|\mathcal{I}|} \sum_{i\in \mathcal{I}_{ca}(1)\cup \mathcal{I}_{ca}(0)} \frac{1}{M_i}\sum_{j=1}^{M_i} \delta_{s^*(\widetilde{B}_{ij}, \widetilde{C}_{ij})} + \delta_{+\infty}\right)\right\} \ge 1-\gamma.
\end{align*}
By the definition of $\widehat{F}^*$, we have
\begin{align*}
   \textup{Quantile}_{1-\gamma}\left(\frac{1}{|\mathcal{I}|} \sum_{i\in \mathcal{I}_{ca}(1)\cup \mathcal{I}_{ca}(0)} \frac{1}{M_i}\sum_{j=1}^{M_i} \delta_{s^*(\widetilde{B}_{ij}, \widetilde{C}_{ij})} + \delta_{+\infty}\right) =  \textup{Quantile}_{1-\gamma}(\widehat{F}^*) = \widehat{q}^*_{1-\gamma},
\end{align*}
we obtain 
\begin{align*}
 P\{\widetilde{C}_{{m+1,1}} \subset \widetilde{\textup{C}}_C(\widetilde{B}_{{m+1,1}})\} &= P\bigg\{\widehat{m}^L(\widetilde{B}_{{m+1,1}})-\widehat{q}^*_{1-\gamma} \le C_{{m+1,1}}^L, C_{{m+1,1}}^R \le \widehat{m}^R(\widetilde{B}_{{m+1,1}})+\widehat{q}^*_{1-\gamma}\bigg\} \\
 &= P\{\max\{\widehat{m}^L(\widetilde{B}_{{m+1,1}})-C_{{m+1,1}}^L, C_{{m+1,1}}^R - \widehat{m}^R(\widetilde{B}_{{m+1,1}})\} \le \widehat{q}^*_{1-\gamma}\} \\
 &= P(s^*(\widetilde{B}_{{m+1,1}},\widetilde{C}_{{m+1,1}}) \le \widehat{q}^*_{1-\gamma})\\
 &\ge 1-\gamma. 
\end{align*}
Combined with the fact that $P\{\widetilde{Y}_{{m+1,1}}(1)-\widetilde{Y}_{{m+1,1}}(0)\in C_{{m+1,1}}\} \ge 1-\gamma$, we have
\begin{align*}
    &P\{\widetilde{Y}_{{m+1,1}}(1)-\widetilde{Y}_{{m+1,1}}(0)\not\in \widetilde{\textup{C}}_C(\widetilde{B}_{{m+1,1}})\} \\
    &=P\{\widetilde{Y}_{{m+1,1}}(1)-\widetilde{Y}_{{m+1,1}}(0)\not\in \widetilde{\textup{C}}_C(\widetilde{B}_{{m+1,1}}), \widetilde{C}_{{m+1,1}} \subset \widetilde{\textup{C}}_C(\widetilde{B}_{{m+1,1}})\} \\
    &\quad + P\{\widetilde{Y}_{{m+1,1}}(1)-\widetilde{Y}_{{m+1,1}}(0)\not\in \widetilde{\textup{C}}_C(\widetilde{B}_{{m+1,1}}), \widetilde{C}_{{m+1,1}} \not\subset \widetilde{\textup{C}}_C(\widetilde{B}_{{m+1,1}})\} \\
    &\le P\{\widetilde{Y}_{{m+1,1}}(1)-\widetilde{Y}_{{m+1,1}}(0)\not\in \widetilde{C}_{{m+1,1}}\} + P\{\widetilde{C}_{{m+1,1}} \not\subset \widetilde{\textup{C}}_C(\widetilde{B}_{{m+1,1}})\} \\
    &\le \alpha + \gamma.
\end{align*}
Finally, by Lemma~\ref{lemma1} (c), $P\{\widetilde{Y}_{{m+1,1}}(1)-\widetilde{Y}_{{m+1,1}}(0)\not\in \widetilde{\textup{C}}_C(\widetilde{B}_{{m+1,1}})\} = P\{Y_{\textup{test}}(1)-Y_{\textup{test}}(0)\not\in \widetilde{\textup{C}}_C(B_{\textup{test}})|B_{\textup{test}} \in \Omega_I\}$, which completes the proof.
\end{proof}

\section{Additional numerical results}\label{sec: additional-simulations}

\subsection{Simulations with 30 clusters}

When $m=30$, the simulation setting is the same as Section 5.1 of the main paper. We provide the simulation results in Table~\ref{tab:sim}. Due to limited clusters, we do not implement the ``B-nested'' method (which requires two sample splittings) or perform the covariate-subgroup analysis for the cluster-level treatment effect (which drops 40\% of all clusters).  For the remaining scenarios, the ``O'' and ``B-direct'' methods can still achieve the target coverage probability, confirming their finite-sample validity. 
However, they become more conservative compared to $m=100$, as reflected by the increased length of intervals. This is expected given limited clusters used in model training and quantile estimation. 
An ad-hoc solution is to consider an 80\% coverage probability, in which case the resulting conformal intervals have a reduced length and are likely more informative. Finally, it is worth highlighting that, even with few clusters, we are able to perform subgroup analysis for the individual-level treatment effect, and its empirical performance is similar to the marginal analysis that uses all data. This observation confirms the statistical benefit of conformal causal inference for individual-level treatment effects in small CRTs. 

\begin{table}[htbp]
\caption{Summary of simulation results 
for 80\% and 90\% conformal intervals 
given $m=30$ clusters. 
For both coverage probability and length of intervals, we report their averages and standard errors. The local treatment effect is conditioned on $|X_{ij1}| < 0.5$. }\label{tab:sim}
\resizebox{\textwidth}{!}{
\begin{tabular}{rrrrrr}
  \hline
\multirow{3}{*}{Treatment effects}  & \multirow{3}{*}{Methods} & \multicolumn{2}{c}{$\alpha=0.2$} & \multicolumn{2}{c}{$\alpha=0.1$}\\
  \cmidrule(lr){3-4} \cmidrule(lr){5-6}
 & & \shortstack[r]{Coverage\\ probability} & \shortstack[r]{Length of\\ intervals}  & \shortstack[r]{Coverage\\ probability} & \shortstack[r]{Length of\\ intervals} \\ 
  \hline
\multirow{2}{*}{Marginal cluster-level} &O & 0.816(0.094) & 3.119(1.465) & 0.912(0.072) & 4.211(2.230) \\ 
 & B-direct & 0.998(0.008) & 5.891(2.833)& 1.000(0.002) & 7.970(4.447) \\ 
  \hline
\multirow{2}{*}{Marginal individual-level}  & O & 0.868(0.035) & 4.311(0.489) & 0.977(0.014) & 6.637(0.794) \\ 
 & B-direct & 0.989(0.007) & 8.345(0.768) & 1.000(0.000) & 12.905(1.232) \\ 
  \hline
\multirow{2}{*}{Local individual-level} & O & 0.870(0.048) & 4.337(0.675) & 0.977(0.021) & 6.933(1.260) \\ 
 & B-direct & 0.988(0.009) & 8.398(1.138) & 1.000(0.000) & 13.542(2.114) \\ 
   \hline
\end{tabular}
}
\end{table}

\subsection{Simulations with 500 clusters}

The simulations in Section of the main paper are summarized in Figures~\ref{fig: 500-cluster} and \ref{fig: 500-individual} below.

\begin{figure}[htbp]
    \centering
    \includegraphics[width=0.95\textwidth]{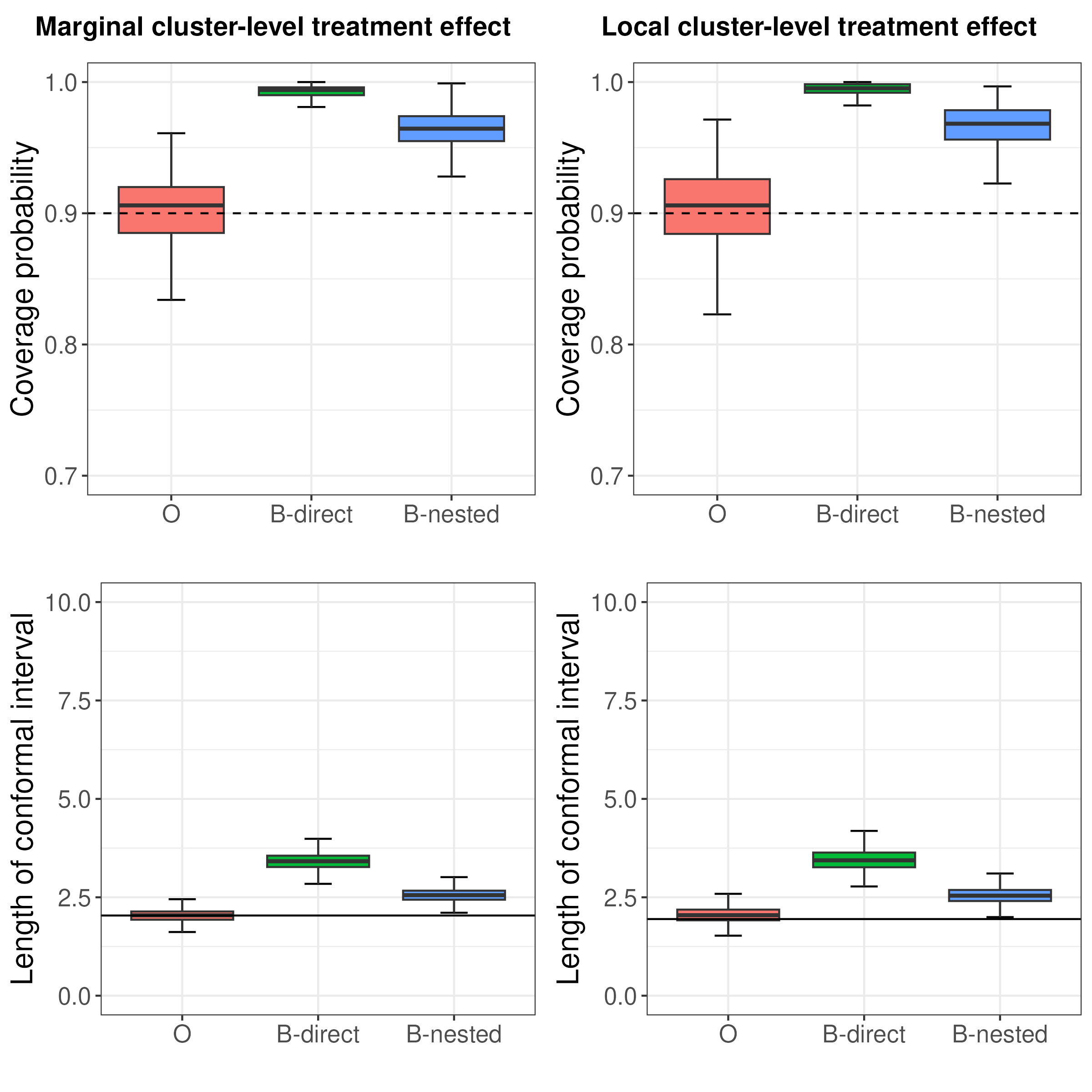} 
    \caption{Simulation results (boxplot) for the marginal (left column) and local (right column, conditioning on $\{R_{i1}\ge 2, R_{i2}=1\}$) cluster-level treatment effects with $m=500$. In the upper panels, the dashed line is the target 90\% coverage probability. In the lower panels, the solid line is the oracle length of conformal intervals, computed as the average length between the $(\alpha/2,1-\alpha/2)$-quantiles of $\overline{Y}(1)-\overline{Y}(0)$ among test data.}
    \label{fig: 500-cluster}
\end{figure}

\begin{figure}[htbp]
    \centering
    \includegraphics[width=0.95\textwidth]{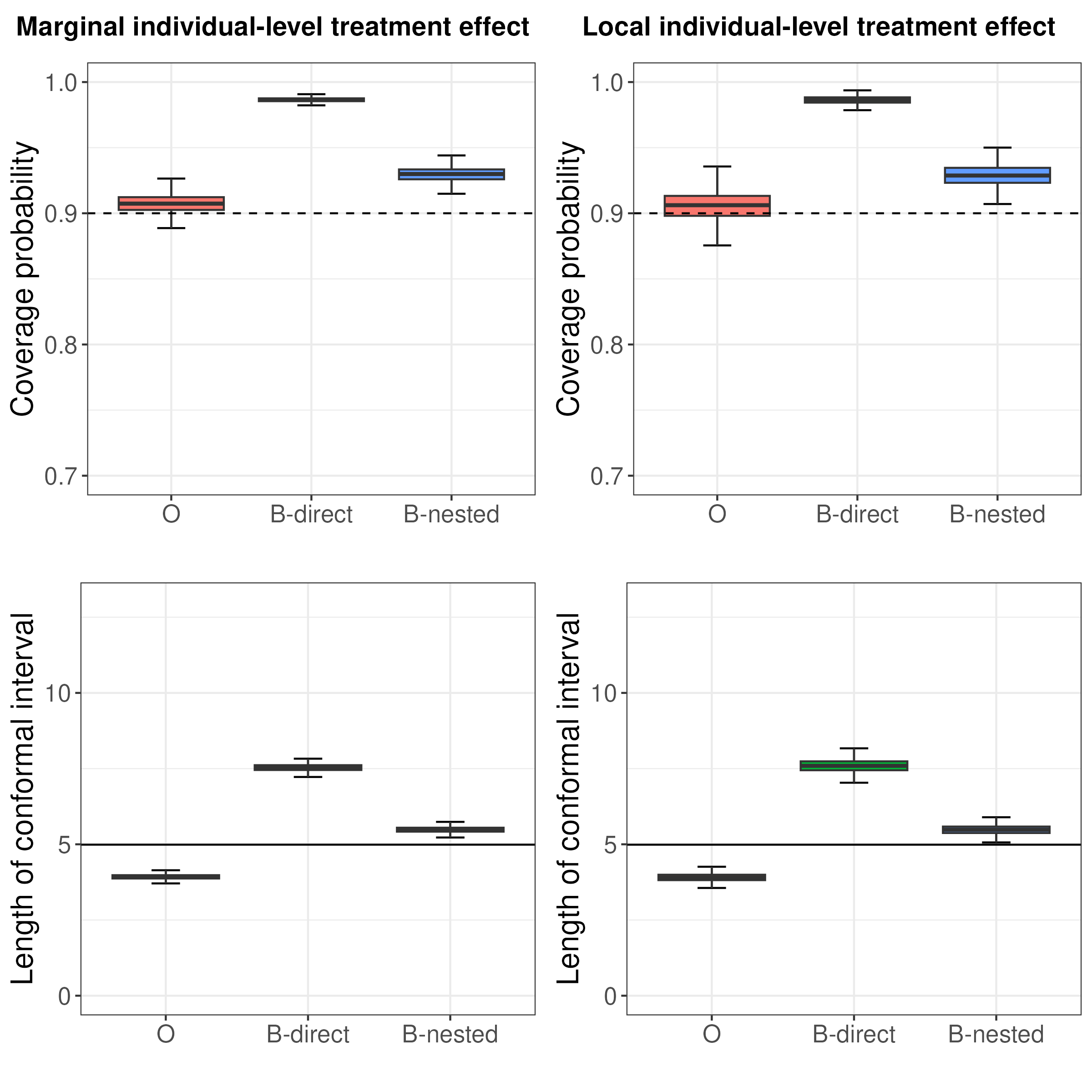} 
    \caption{Simulation results (boxplot) for the marginal (left column) and local (right column, conditioning on $\{|X_{ij1}|<0.5\}$) individual-level treatment effects with $m=500$. In the upper panels, the dashed line is the target 90\% coverage probability. In the lower panels, the solid line is the oracle length of conformal intervals, computed as the average length between the $(\alpha/2,1-\alpha/2)$-quantiles of $Y(1)-Y(0)$ among test data.}
    \label{fig: 500-individual}
\end{figure}

\subsection{Simulations with linear regression as the prediction model}

The simulation results with linear regression as the prediction model is given in Table~\ref{fig: sim-linear-regression}. We compare it with Table~\ref{tab:sim}, which uses linear regression and random forest as the prediction model. We observe a short length of conformal intervals when using random forest.

\begin{table}[htbp]
\caption{Summary of simulation results 
for 80\% and 90\% conformal intervals 
given $m=30$ clusters with linear regression as the prediction model. 
For both coverage probability and length of intervals, we report their averages and standard errors. The local treatment effect is conditioned on $|X_{ij1}| < 0.5$. } \label{fig: sim-linear-regression}
\centering
\resizebox{\textwidth}{!}{
\begin{tabular}{rrrrrr}
  \hline
\multirow{3}{*}{Treatment effects}  & \multirow{3}{*}{Methods} & \multicolumn{2}{c}{$\alpha=0.2$} & \multicolumn{2}{c}{$\alpha=0.1$}\\
  \cmidrule(lr){3-4} \cmidrule(lr){5-6}
 & & \shortstack[r]{Coverage\\ probability} & \shortstack[r]{Length of\\ intervals}  & \shortstack[r]{Coverage\\ probability} & \shortstack[r]{Length of\\ intervals} \\ 
  \hline
\multirow{2}{*}{Marginal cluster-level} &O & 0.821(0.095) & 5.149(4.266) & 0.915(0.071) & 7.457(8.898) \\ 
 & B-direct & 0.999(0.007) & 10.310(8.526)& 1.000(0.001) & 14.910(8.017) \\ 
  \hline
\multirow{2}{*}{Marginal individual-level}  & O & 0.865(0.044) & 4.966(1.114) & 0.972(0.020) & 7.636(1.636) \\ 
 & B-direct & 0.994(0.006) & 9.706(2.076) & 1.000(0.000) & 15.027(3.010) \\ 
  \hline
\multirow{2}{*}{Local individual-level} & O & 0.861(0.057) & 4.722(1.067) & 0.974(0.025) & 7.556(1.725) \\ 
 & B-direct & 0.991(0.009) & 9.229(1.877) & 1.000(0.000) & 14.700(2.893) \\ 
   \hline
\end{tabular}
}
\end{table}

\subsection{Simulation for cluster treatment effect based on average of individual-level prediction models}

As mentioned in Section 3.1 of the main paper, the prediction model in Algorihtm 1 can be replaced by the average of individual-level predictions. Here,  we compare this approach with our original proposal based on cluster-level data in our simulations, and Figure~\ref{fig:sim-comparison} shows that this new approach can slightly narrow the length of conformal intervals on average while maintaining validity. 
    \begin{figure}
        \centering
\includegraphics[width=0.9\linewidth]{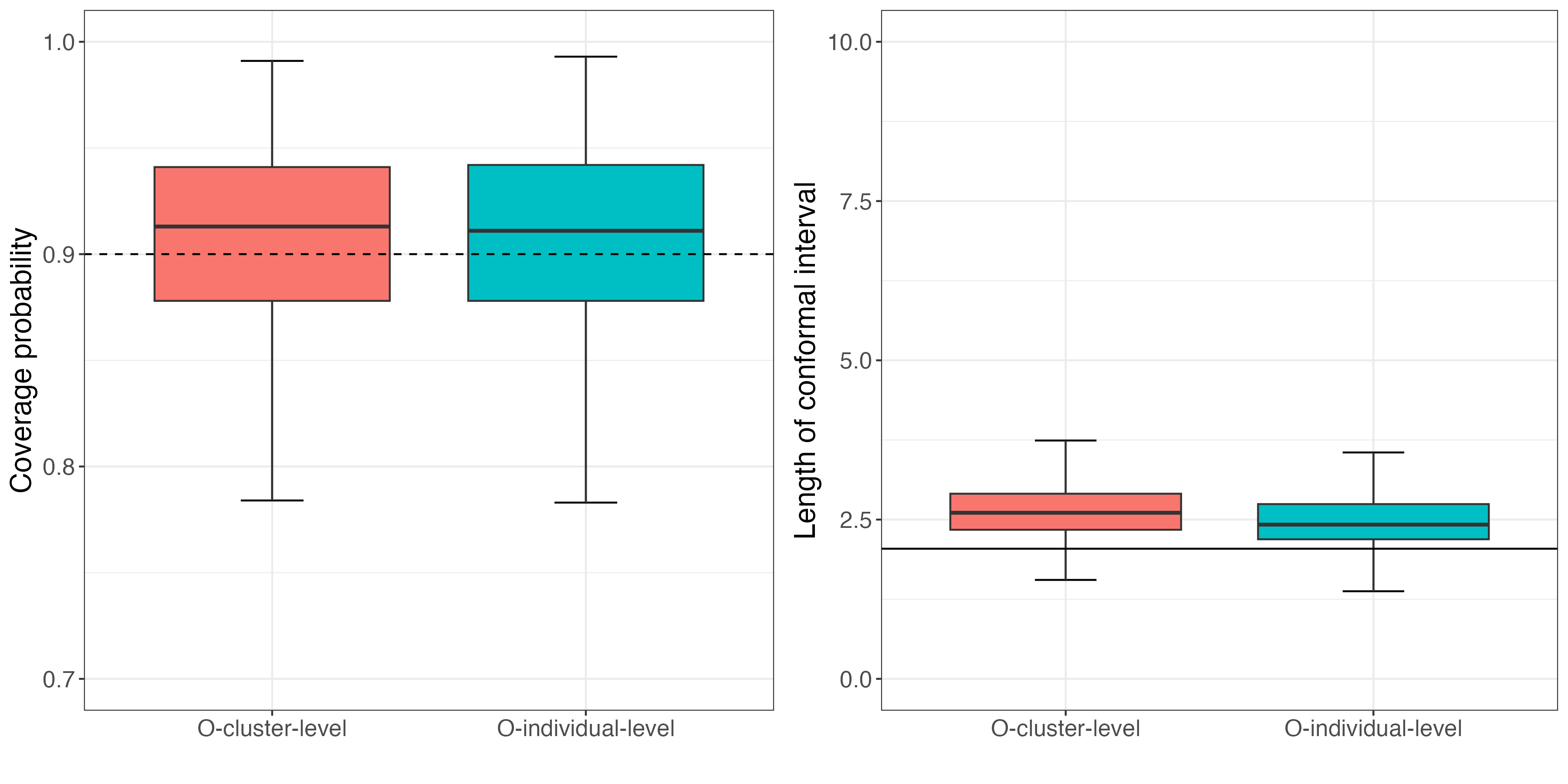}
        \caption{Comparison of prediction models based on cluster-level data (``O-cluster-level'') versus individual-level data (``O-individual-level'') in terms of coverage probability (left panel) and length of conformal intervals (right panel).}
        \label{fig:sim-comparison}
    \end{figure}

\subsection{Data application for local treatment effects}

In addition to the marginal analysis in the main paper, we performed subgroup analyses on individuals with severe baseline pain (baseline PEGS score equal to or larger than 7, 40\% participants) and moderate baseline pain (baseline PEGS score between 4 and 7, 50\% participants). We summarize the results in Table~\ref{data-application2}, which shows no difference in the length of intervals or the fraction of negatives between the two subgroups. Overall, these analyses indicate treatment benefits for a small subset of the clusters/individuals, which drive the overall treatment effect signal provided in previous analyses \citep{debar2022primary,wang2023CRT}. 

\begin{table}[htbp]
\caption{Summary results of data application for individual-level treatment effects on covariate subgroups. For both the length of intervals and the fraction of negatives, we present the average and standard error over 100 runs.}\label{data-application2} 
\centering
\resizebox{\textwidth}{!}{
\begin{tabular}{ccccc}
  \hline
\multirow{2}[3]{*}{\shortstack{Coverage\\ probability}}  & \multicolumn{2}{c}{Subgroup: severe baseline pain} & \multicolumn{2}{c}{Subgroup: moderate baseline pain}\\
  \cmidrule(lr){2-3} \cmidrule(lr){4-5}
 & Length of intervals & Fraction of negatives & Length of intervals & Fraction of negatives \\ 
  \hline
90\% & 6.538(0.571) & 0.072(0.042) & 6.508(0.538) & 0.076(0.042) \\ 
  80\% & 4.982(0.465) & 0.120(0.051)  & 4.928(0.417) & 0.132(0.053) \\ 
  70\% & 3.935(0.353) & 0.180(0.064) & 3.960(0.349) & 0.185(0.068) \\ 
  60\% & 3.180(0.262) & 0.243(0.067) & 3.221(0.296) & 0.241(0.079) \\ 
   \hline
\end{tabular}
}
\end{table}

\section{Extensions to cluster observational studies}\label{sec: cluster-observational-studies}
      We focus on clustered observational studies with a cluster-level treatment assignment, and extend our development to accommodate a non-randomized treatment assignment mechanism. In clustered observational studies, just like the CRT setting, individual observations within a cluster are typically correlated. Denoting $[m]$ as the set $\{1,\dots,m\}$, the ``strong ignorability" condition in this context can be expressed as follows: for any cluster $i\in [m]$, the set $\{(Y_{i,1}(0),Y_{i,1}(1)),\cdots, (Y_{i,N_i}(0),Y_{i,N_i}(1))\}$ is independent of $A_{i}$ given $\textbf{B}_i:=(X_{i,1},\cdots,X_{i,N_i},R_i).$ Here, $R_i$ denotes the cluster-level covariates and $X_{i,j}, j\in [N_i]$ represent individual-level covariates.)

The joint distribution of $\textbf{Y(0)}, \textbf{B}$ among the treated, is $P_{\textbf{Y(0)}, \textbf{B}|A=1}=P_{\textbf{Y(0)}|\textbf{B}}\cdot P_{\textbf{B}|A=1} $, where $\textbf{Y(0)}:=(Y_{\bullet,1}(0),\cdots, Y_{\bullet,N_{\bullet}}(0)),$ and $\textbf{B}:=(X_{\bullet,1},\cdots,X_{\bullet,N_{\bullet}},R_{\bullet}).$
For the data $\textbf{Y(0)}$ already observed in the current study, the data are collected under the distribution $P_{\textbf{Y(0)}, \textbf{B}|A=0}=P_{\textbf{Y(0)}|\textbf{B}}\cdot P_{\textbf{B}|A=0} $. Consequently, there is a covariate shift in the target distribution.

We next describe our inferential targets. Suppose we have newly observed covariates \( \textbf{B}_{test} \) in the treatment group (\( A = 1 \)). At the cluster level, with a pre-specified level $1-\alpha$, our objective is to construct a prediction interval \(\widetilde{C}(\overline{B}_{test})\) that contains the unobserved potential outcome \(\overline{Y}_{test}(0)\) with probability $1-\alpha$, 
i.e., $$P\{\overline{Y}_{test}(0) \in \widetilde{C}(\overline{B}_{test})|A=1\} \ge 1-\alpha.$$ 
If $\overline{Y}(1)$ is directly observed, the prediction interval for the treatment effect is then $\overline{Y}(1)-\widetilde{C}(\overline{B}_{test}).$
On the other hand, if \(\overline{Y}(1)\) is not observed, we can still construct a prediction interval for \(\overline{Y}(1) - \overline{Y}(0)\) by leveraging the fact that \(\overline{Y}(1)\) is sampled from the same distribution as the target distribution (i.e., there is no covariate shift). We apply the method described in the main text to construct a prediction interval \(\widetilde{C}_{C,1}(\overline{B}_{test})\) with level \(1-\alpha\) for \(\overline{Y}(1)\). We then combine this prediction interval with the one for \(\overline{Y}(0)\) (also at level \(1-\alpha\)) to achieve a coverage guarantee of \(1-2\alpha\) for \(\overline{Y}(1) - \overline{Y}(0)\), following a similar approach to that outlined in equation (6) of the main text.

Similarly, at the individual level, we aim to construct a prediction interval $\widetilde{C}(B_{test})$ for \( Y(0) \), the unobserved individual-level potential outcome, with coverage probability $1-\alpha$, i.e., $$P(Y(0) \in \widetilde{C}(B_{test})|A=1)\ge 1-\alpha.$$ If \(Y(0)\) is directly observed, the prediction interval for the treatment effect is given by \(Y(1) - \widetilde{C}({B_{test}})\). If \(Y(1)\) is not directly observed, the method for constructing a prediction interval with level \(1-2\alpha\) is nearly identical to the approach used at the cluster level. The key difference is that we construct an individual prediction interval \(\widetilde{C}_{I,1}(B_{\text{test}})\) with level \(1-\alpha\) for \(Y(1)\) using the approach described in the main text.


Next, we use the technique mentioned in \cite{tibshirani2019conformal} to construct prediction intervals for $\overline{Y}(0)$ and $Y(0)$, to achieve the aforementioned inferential targets. 

    \begin{itemize}
        \item At the cluster level, since each cluster is sampled i.i.d. from the cluster super-population, the technique for constructing the prediction for $\overline{Y}(0)$ is almost identical with the case studied in \cite{tibshirani2019conformal}, except that we replace the individual-level covariates in \citet{tibshirani2019conformal} with covariates from the whole cluster $\textbf{B}$ in constructing the weight function (described below).  Hence, when covariate shift exists, we can directly follow the method proposed in \cite{tibshirani2019conformal} by constructing the prediction set as 
        \begin{align*}
        \hat C_m(\overline{B})=\bigg\{y\in \mathbb{R}: S{(\overline{B},y)}\le Q_{1-\alpha}\Big(\sum_{i=1}^m p_i^{w}(\textbf{B}_i)\delta_{S{(\overline{B}_i,y_i)}}+p^{w}_{m+1}(\textbf{B})\delta_{+\infty}\Big)\bigg\}.
        \end{align*}
       Here we let 
       \begin{align*}
        p_i^{w}(\textbf{B}_i)=\frac{w(\textbf{B}_i)}{\sum_{j=1}^m w(\textbf{B}_j)+w(\textbf{B})},~~~\text{and}~~~ p_{m+1}^w=\frac{w(\textbf{B})}{\sum_{j=1}^m w(\textbf{B}_j)+w(\textbf{B})},  
       \end{align*}
       where the weight function $w(\textbf{B}):=d{P}_{\textbf{B}|A=1}(\textbf{B})/dP_{\textbf{B}|A=0}(\textbf{B}).$ Note that we utilize \(\textbf{B}\) in the weight function instead of \(\overline{B}\) because the aforementioned strong ignorability assumption is conditioned on \(\textbf{B}\) rather than on \(\overline{B}\). Additionally, $S{(\overline{B}_i,y_i)},i\in[m+1],$ is the non-conformity score of the $i$-th cluster with $A=0.$ Notice that here there is a slight abuse of notation as we let $y_i$ be the response that we are interested in. For example, in the above, $y_i $ can be $\overline{Y}_i(0).$ Additionally, without loss of generality, we also assume the non-conformity score function $S(\cdot,\cdot)$ is pre-trained via sample splitting.
       \item 
At the individual level, constructing prediction sets becomes more challenging due to the violation of the original exchangeability assumption under this clustered observational data structure. To address this issue, we extend the generalized conformal prediction method from \cite{tibshirani2019conformal} to hierarchically structured data. In this discussion, we focus on the case where all clusters have equal cluster sizes, denoted by \( M > 0 \). We will address the case of clusters with random sizes in future work.

    Suppose that the first $n$ clusters are involved in the trial data with treatment $A=0$ and the $(n+1)$-th one is a new cluster with treatment $A=1$. We denote $C_i=(s_{i1},\cdots,s_{iM}), i\in [n+1]$ as a collection of non-conformity scores (associated with $Y(0)$ and $B$) in the $i$-th cluster and $g(s_{i1},\cdots s_{iM})$ as the joint density function of $s_{i1},\cdots s_{iM}|A=0$. Then,  we can write the joint distribution of non-conformity scores among these $n+1$ clusters to be $$f(C_1,\cdots C_{n+1})=\Pi_{i=1}^{n+1}h_i(\mathbf{B}_i)g(s_{i1},\cdots,s_{iM}),$$ where $h_i(\mathbf{B}_i) = 1$ for $i\in[n]$ and $h_{n+1}(\mathbf{B}_{n+1})= w(\mathbf{B}_{n+1})=d{P}_{\mathbf{B}|A=1}(\mathbf{B}_{n+1})/dP_{\mathbf{B}|A=0}(\mathbf{B}_{n+1})$. 
    Here, because clusters are independent of each other, we are able to decompose $f(C_1,\cdots C_{n+1})$ into products of cluster-specific densities. Since the first $n$ clusters are collected under $A_i=0$, no reweighting is needed, which makes $h_i(\mathbf{B}_i) = 1$ for $i\in[n]$. However, since the $(n+1)$-th cluster is observed under $A_i=1$, there is a covariate shift and we need to apply propensity score weighting.
    
     Next, under the within-cluster exchangeability condition, we further have $$g(s_{i\sigma(1)},\cdots,s_{i\sigma(M)})=g(s_{i1},\cdots,s_{iM}),~~~\text{for all}~~~i\in [n+1],$$ 
     where $\sigma(1),\dots, \sigma(M)$ is any permutation of indices  $1,\dots, M$. 
Let $E_z$ denote the observed event that $$\{\{S_{1,1},\cdots,S_{1,M}\},\cdots, \{S_{n+1,1},\cdots,S_{n+1,M}\}\}=\{\{s_{1,1},\cdots,s_{1,M}\},\cdots, \{s_{n+1,1},\cdots,s_{n+1,M}\}\}.$$ 
Since the dataset can be labeled in any order, this equality only assume these two sets are the same, regardless of the order.      
We further define \\ $\{\sigma(1,1),\dots, \sigma(1,M),\dots, \sigma(n+1,1), \dots, \sigma(n+1,M)\}$ as a permutation of indices $\{(1,1),\dots, (1,M),\dots, (n+1,1), \dots, (n+1,M)\}$ by first permuting the cluster indices (the first entry) and then the individual indices (the second entry).
Following the proof procedure of Lemma 3 in \cite{tibshirani2019conformal} but in a hierarchical exchangeable setting, we obtain
       \begin{align*}
       P(S_{n+1,1}=s_{i,j}|E_z)&=\frac{\sum_{\sigma:\sigma(n+1,1)=(i,j)}f(\{s_{\sigma(1,1)},\cdots,s_{\sigma(1,M)}\},\cdots, \{s_{\sigma(n+1,1)},\cdots,s_{\sigma(n+1,M)}\})}{\sum_{\sigma}f(\{s_{\sigma(1,1)},\cdots,s_{\sigma(1,M)}\},\cdots, \{s_{\sigma(n+1,1)},\cdots,s_{\sigma(n+1,M)}\})}\\&=\frac{\sum_{\sigma:\sigma(n+1,1)=(i,j)}\Pi_{i=1}^{n+1} h_i(\mathbf{B}_{\sigma(i)})g(\{s_{\sigma(i,1)},\cdots,s_{\sigma(i,M)}\})}{\sum_{\sigma}\Pi_{i=1}^{n+1} h_i(\mathbf{B}_i)g(\{s_{\sigma(i,1)},\cdots,s_{\sigma(i,M)}\})}\\&=\frac{w(B_i)}{M\cdot w(B_1)+\cdots+ M\cdot w(B_n)+M\cdot w(B_{n+1})}:=p_{i,j}^w.
       \end{align*}
       The second equality follows from the fact that $h_i(\cdot)=1$ for all $i\in [n]$ and the factors involving $g$ are all cancelled. 
       Therefore, the final prediction set with coverage level $1-\alpha$ is constructed as 
       \begin{align*}
        \widetilde{C}_n(B_{n+1,1})=\bigg\{y\in R: S(B_{n+1,1})\le Q_{1-\alpha}\Big(\sum_{i=1}^n\sum_{j=1}^M p_{i,j}^{w}\delta_{S(B_{i,j},y_{i,j})}+\sum_{j=1}^M p^{w}_{n+1,j}\delta_{+\infty}\Big)\bigg\}.
       \end{align*}
    \end{itemize}
    
 \bibliographystyle{apalike}
 \bibliography{ref}

